\documentclass[a4paper]{article}

\usepackage{amsmath,amsthm}
\usepackage{amssymb,latexsym}
\usepackage{comment}
\usepackage{amsfonts}
\usepackage{tikz}
\usepackage[utf8]{inputenc}
\usepackage[english]{babel}
\usepackage[all]{xy}
\usepackage{fullpage}
\usepgflibrary{arrows}

\newcommand\norm[1]{\|#1\|}
\newcommand\bary[1]{\mathrm{bar}(#1)}
\newcommand\dm{d_{\mu,m}}
\newcommand\dPm{d_{\mu,m}^P}
\newcommand\dmP{d_{\mu_P,m}}
\newcommand\dPP{d_{\mu_P,m}^P}
\newcommand\dPW{d_{\mu_P,m}^W}
\newcommand\X{\mathbb{X}}
\newcommand\R{\mathbb{R}}
\newcommand\U{\mathbb{U}}
\newcommand\V{\mathbb{V}}
\newcommand\W{\mathbb{W}}
\newcommand\dX[2]{d_\X(#1,#2)}

\newcommand\Dgm[1]{\mathrm{Dgm}(#1)}
\newcommand\Sub[2]{\mathrm{Sub}_#1(#2)}
\newcommand\Supp[1]{\mathrm{Supp}(#1)}
\newcommand\dy{\mathrm{d}y}
\newcommand\dbl{d_B^{\mathrm{log}}}

\newcommand\ForAuthors[1]
 {\par\smallskip                     
  \begin{center}
   \fbox
   {\parbox{0.9\linewidth}
    {\raggedright\sc--- #1}
   }
  \end{center}
  \par\smallskip                     %
 }

\newtheorem{theorem}{Theorem}[section]
\newtheorem{definition}[theorem]{Definition}
\newtheorem{lemma}[theorem]{Lemma}
\newtheorem{corollary}[theorem]{Corollary}
\newtheorem{proposition}[theorem]{Proposition}
\newtheorem{remark}{Remark}

\newcommand{\e}{\varepsilon}

\newcommand{\ir}{\lambda}
\newcommand{\net}{N}
\newcommand{\cl}[1]{\overline{#1}}
\newcommand{\ff}{\mathrm{f}}
\newcommand{\proj}{\pi}

\newcommand{\Hom}{\mathrm{H_*}}
\DeclareMathOperator\argmin{argmin}
\DeclareMathOperator\argmax{argmax}

\title{Efficient and Robust Persistent Homology for Measures}
\author{
      Micka\"{e}l Buchet\footnote{mickael.buchet@inria.fr}
    \and
      Fr\'{e}d\'{e}ric Chazal\footnote{frederic.chazal@inria.fr}
    \and 
      Steve Y. Oudot\footnote{steve.oudot@inria.fr}
    \and
      Donald R. Sheehy\footnote{don.r.sheehy@gmail.com}
  }

\begin{document}
\maketitle

  A new paradigm for point cloud data analysis has emerged recently, where point clouds are no longer treated as mere compact sets but rather as empirical measures. 
  A notion of distance to such measures has been defined and shown to be stable with respect to perturbations of the measure. 
  This distance can easily be computed pointwise in the case of a point cloud, but its sublevel-sets, which carry the geometric information about the measure, remain hard to compute or approximate.
  This makes it challenging to adapt many powerful techniques based on the Euclidean distance to a point cloud to the more general setting of the distance to a measure on a metric space.
 
  We propose an efficient and reliable scheme to approximate the topological structure of the family of sublevel-sets of the distance to a measure. 
  We obtain an algorithm for approximating the persistent homology of the distance to an empirical measure that works in arbitrary metric spaces. 
  Precise quality and complexity guarantees are given with a discussion on the behavior of our approach in practice.

\section{Introduction} 
  Given a sample of points $P$ from a metric space $\X$, the distance function $d_P$ maps each $x\in \X$ to the distance from $x$ to the nearest point of $P$.
  The related fields of geometric inference and topological data analysis have provided a host of theorems about what information can be extracted from the distance function, with a particular focus on discovering and quantifying intrinsic properties of the shape underlying a data set~\cite{stcsesCCL,fhshcrsNSW}.
  The flagship tool in topological data analysis is persistent homology and the most common goal is to apply the persistence algorithm to distance functions, either in Euclidean space or in metric spaces~\cite{carlsson09topology,tpsELZ, cphCZ}.
  From the very beginning, this line of research encountered two major challenges.
  First, distance functions are very sensitive to noise and outliers (Fig.~\ref{fig:d_mu} left).
  Second, the representations of the sublevel sets of a distance function become prohibitively large even for moderately sized data.
  These two challenges led to two distinct research directions.
  First, the distance to the data set was replaced with a distance to a measure induced by that data set~\cite{gipmCCM}.
  The resulting theory is provably more robust to outliers, but the sublevel sets become even more complex to represent (Fig.~\ref{fig:d_mu} center).
  Towards more efficient representations, several advances in \emph{sparse filtrations} have led to linear-size constructions~\cite{ctpsmDFW,ZZZ,lsavrfS}, but all of these methods exploit the specific structure of the distance function and do not obviously generalize.
  In this paper, we bring these two research directions together by showing how to combine the robustness of the distance to a measure, with the efficiency of sparse filtrations.

  \begin{figure}[htbp]
    \centering
      \includegraphics[width=.3\textwidth]{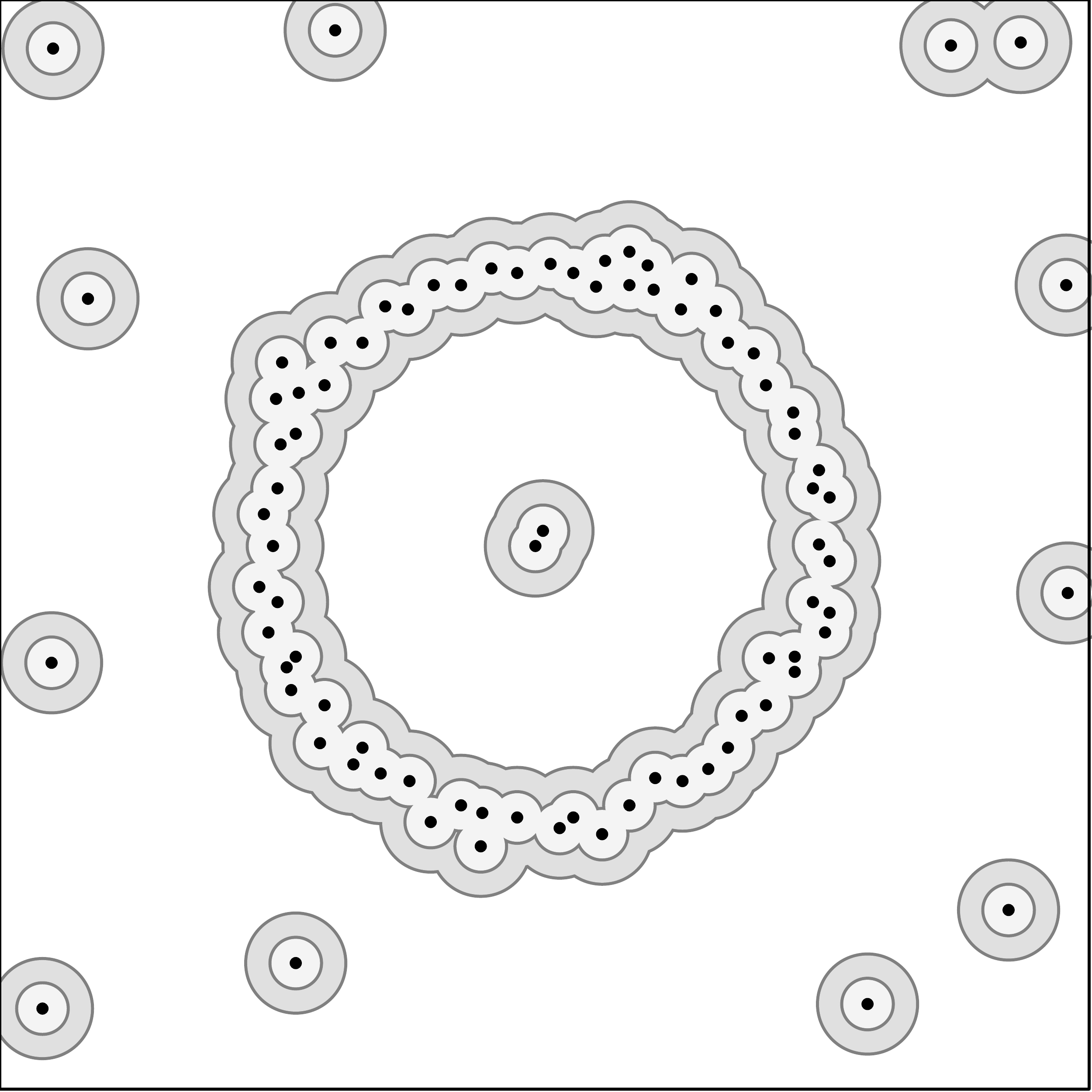}
      \includegraphics[width=.3\textwidth]{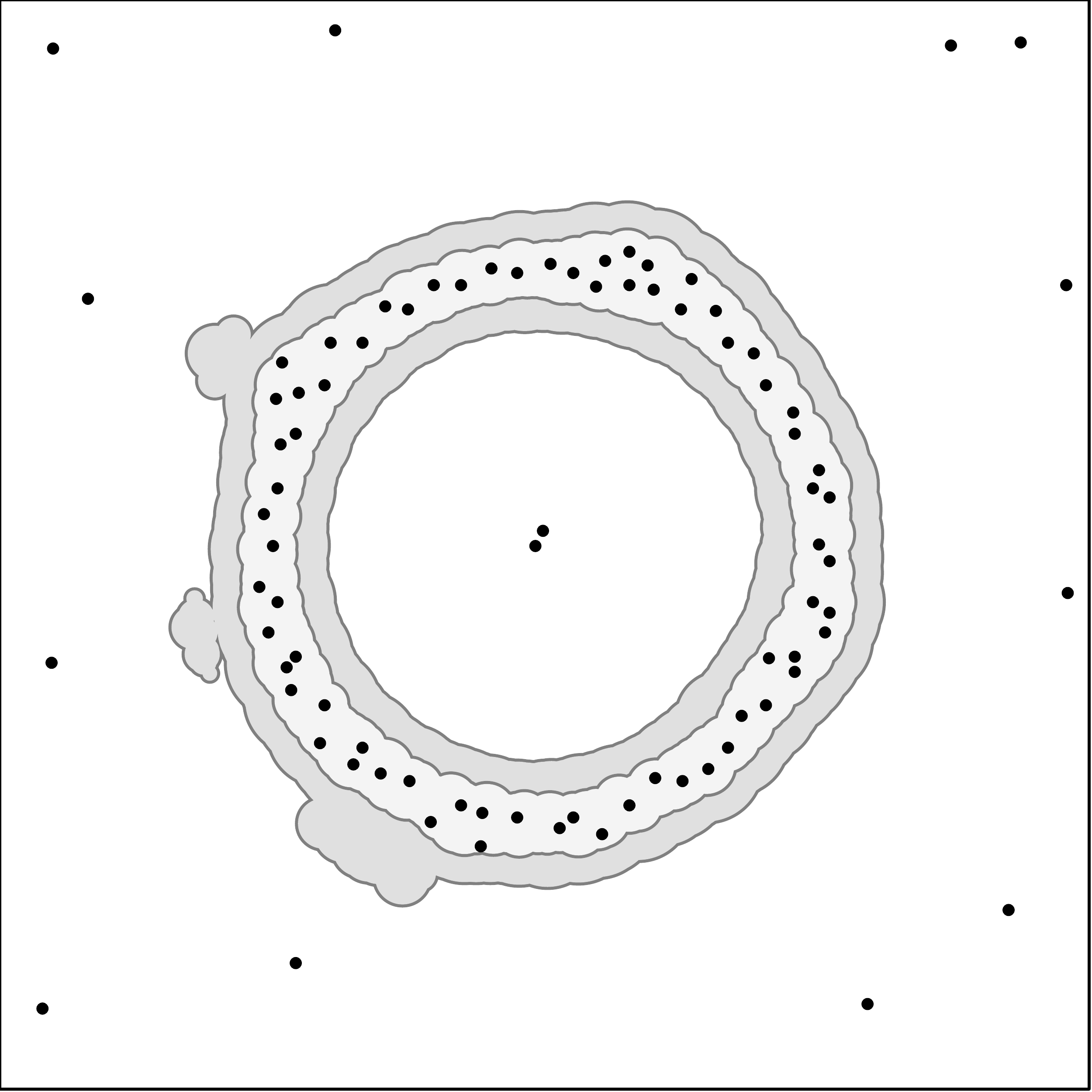}
      \includegraphics[width=.3\textwidth]{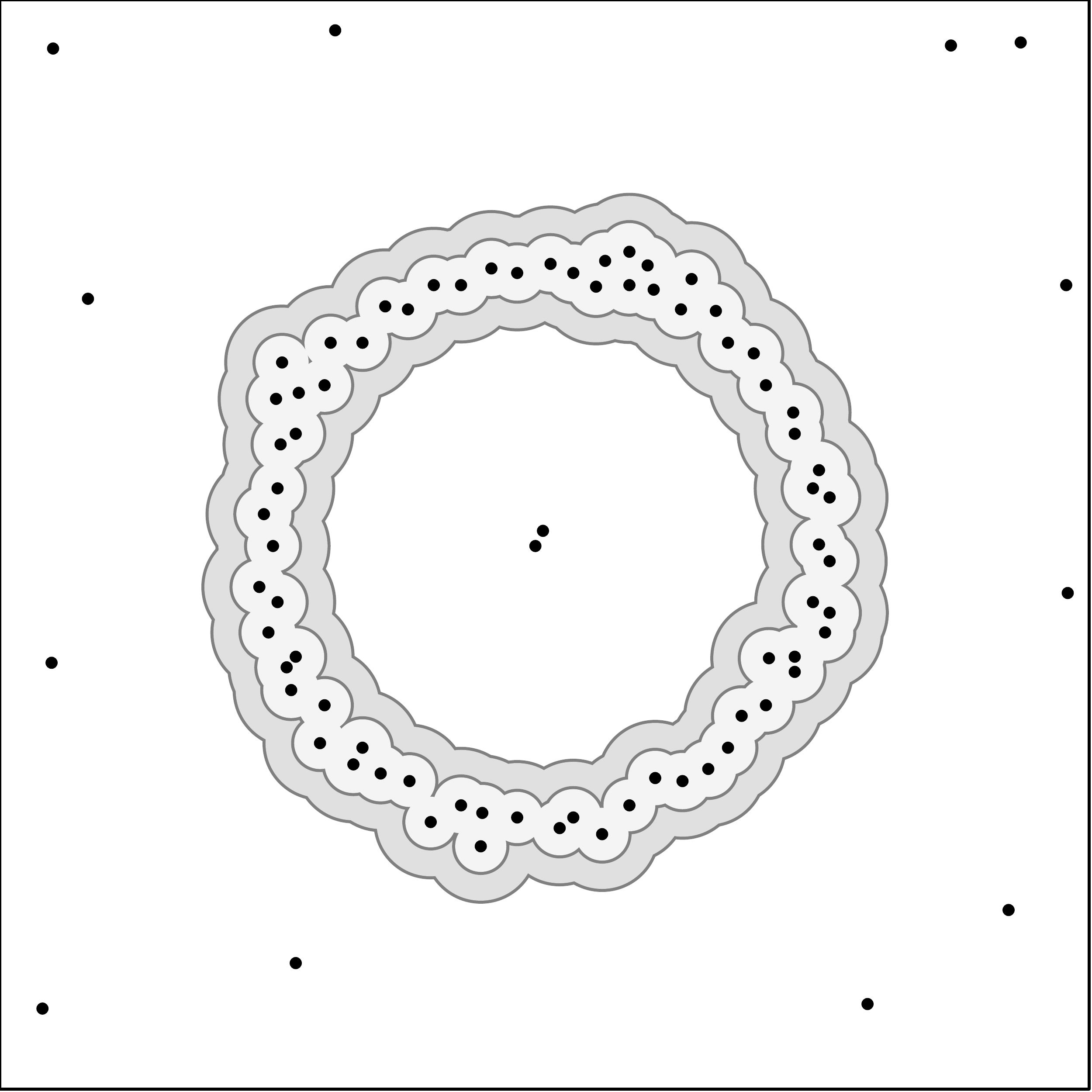}
    \caption{From left to right, two sublevel sets for $d_P$, $\dmP$, and $\dPP$ with $m  = \frac{3}{|P|}$.
    The first is too sensitive to noise and outliers.
    The second is smoother, but substantially more difficult to compute.
    The third is our approximation, which is robust to noise, efficient to compute, and compact to represent.}
    \label{fig:d_mu}
  \end{figure}

\paragraph{Contributions:}
\begin{enumerate}
  \item \textbf{A Generalization of the Wasserstein stability and persistence stability of the distance to a measure for triangulable metric spaces.}
  \item \textbf{A general method for approximating the sublevel sets of the distance to a measure by a union of balls.}  
  Our method uses $O(n)$ balls for inputs of $n$ samples.  
  Known methods for representing the exact sublevel sets can require $n^{\Theta(d)}$ balls. 
  Existing approximations using a linear number of balls are only applicable in Euclidean space~\cite{wkdGMM}.
  \item \textbf{A linear size approximation to the weighted Rips filtration.}
  For intrinsically low-dimensional metric spaces, we construct a filtration of size $O(n)$ that achieves a guaranteed quality approximation.  
  Specifically, if the doubling dimension of the metric is $d$ then the size complexity is $2^{O(dk)}n$ if one considers simplices up to dimension $k$ (see Def.~\ref{def:doubling_dimension} for the formal definition of doubling dimension).
  This is a significant improvement over the full weighted Rips filtration, which has size $2^n$ in general or size $n^{k+1}$ if one considers only simplices up to dimension $k$.
  It also has the advantage that the sparsification is independent of the weights.
  Thus, the (geo)metric preprocessing phase can be reused for any weighting of the points.
  If one attempted to use previous methods directly, this preprocessing phase would have to be repeated for each set of weights.
  This is especially useful if one is interested in several different weight functions such as when approximating the distance to a measure for several different values of the mass parameter.
  \item \textbf{An effective implementation with experimental results.} 
\end{enumerate}

\paragraph{Overview of the paper} 

Originally, the distance to a measure was introduced to capture information about both scale and density in a Euclidean point cloud.
We extend the distance to a measure to any metric space $\X$.
We write $\bar{B}(x,r)$ to denote the closed ball with center $x$ and radius $r$.
The distance to a measure is then defined as follows.
\begin{definition}
Let $\mu$ be a probability measure on a metric space $\X$ and let $m\in]0,1]$ be a mass parameter. We define the distance $\dm$ to the measure $\mu$ as
\[
  \dm:x\in\X\mapsto\sqrt{\frac{1}{m}\int_0^m\delta_{\mu,l}(x)^2 dl},
\]
where $\delta_{\mu,m}$ is defined as
\[
  \delta_{\mu,m}:x\in \X \mapsto \inf\{r>0\mid\mu(\bar{B}(x,r))>m\}.
\]
\end{definition}


The distance to a measure has interesting inference and stability results in the Euclidean setting~\cite{gipmCCM}.
That is, the sublevel sets of the function can be used to infer the topology of the support of the underlying distribution (inference), and also, the output for similar inputs will be similar (stability).
In Section~\ref{sPersStabDtm}, we extend these stability results to any metric space.
The results about the stability of persistence diagrams apply to any triangulable metric space, i.e. metric spaces homeomorphic to a locally finite simplicial complex (the persistence diagram may not exist for non-triangulable metric spaces).

We then give a new way to approximate the distance to a measure.
Using a sampling of the support of a measure, we are able to compute accurately the sublevel sets of the distance to a measure in any metric space, using power distances.
We show in Section~\ref{sApd} that these functions have adequate stability and approximation properties.
Then, in Section~\ref{sRestrict}, we give the practical implications for computing persistence diagram for finite samples.

The \emph{witnessed $k$-distance} is another approach to approximating the distance to a measure proposed in~\cite{wkdGMM}.
This approach works only in Euclidean spaces as it relies on the existence of barycenters of points.
The analysis links the quality of the approximation to the underlying topological structure.
In this paper, we look at bounds independent of intrinsic geometry.
When restricted to the Euclidean setting in section~\ref{ssEc}, our method improves the approximation bounds from~\cite{wkdGMM}.
The new bounds match the quality of approximation achieved by our method of Section~\ref{sApd}, which has the added advantage that it is valid in any metric spaces..

In Section~\ref{sWRips}, we introduce the \emph{weighted Rips complex}.
Given a parameter, the sublevel set of a power distance associated with this parameter is a union of balls.
Generalizing the Vietoris-Rips complex, we define the weighted Rips complex as the clique complex whose $1$-skeleton is the same as the one of the nerve of this union of balls.
The induced filtration has important stability properties and can be used to approximate  persistence diagrams.

Unfortunately, the weighted Rips filtration is too large to construct in full for large instances.
This problem already exists with the usual Rips filtration.
Sparsifying schemes have been recently proposed in~\cite{ctpsmDFW,lsavrfS}.
Extending the approach used in~\cite{lsavrfS}, we construct a sparse approximation that has linear size in the number of points (Section~\ref{sec:sparse_rips}).
This can be used to approximate persistence diagrams even for high dimensional inputs if the data is intrinsically low dimensional.
As we show in Section~\ref{sec:sparse_rips}, there are very simple examples where the input metric is intrinsically low-dimensional and yet the weight function can cause the weighted distance function to be high-dimensional.  
Our approach has the advantage over previous methods in that the size complexity will only depend on the dimension of the input metric, rather than the dimension of points under the weighted distance.

The combination of these approaches makes it possible to use the distance to a measure to infer topology on real instances.
In Section~\ref{sNumeric}, we illustrate the theory with some examples and results from an implementation.

\section{Background} 
In this paper, we consider a metric space $\X$ with distance $\dX{\cdot}{\cdot}$.
In a slight abuse of notation, we also write $d_\X$ to denote the distance between a point and a set defined as $\dX{x}{P}=\inf_{p\in P}\dX{x}{p}$.
The Hausdorff distance between two sets $P$ and $Q$ will be denoted $d_H(P,Q)$.
We write $B(x,r)$ for the open ball of center $x$ and radius $r$ in $d_\X$, and we write $\bar{B}(x,r)$ for the corresponding closed ball.

\paragraph{Metric Spaces and doubling dimension}
For metric spaces that are not embedded in Euclidean space, the doubling dimension gives a useful way to describe the intrinsic dimension of the metric space by bounding the size of certain covers of subsets.
Formally it is defined as follows.

\begin{definition}\label{def:doubling_dimension}
  The \emph{doubling constant} $\lambda_{\X}$ of a metric space $\X$ is the maximum over all balls $B(x,r)$ with $x\in\X$ of the minimum number of balls of radius $r/2$ required to cover $B(x,r)$.
  The \emph{doubling dimension} is defined to be $\log_2(\lambda_{\X})$.
\end{definition}

\paragraph{Wasserstein distance\\}
To compare measures, we use the Wasserstein distance, also called the earth-mover distance.
Intuitively, it is the minimal cost to move all the mass from one measure to another.
To state the formal definition we first introduce some notation.

Given a measure $\mu$ on a metric space $\X$, we write $\mathfrak{B}(\X)$ to denote the set of all Borel subsets of $\X$.
Given $A\in\mathfrak{B}(\X)$, we define the \emph{mass of $A$} as $\mu(A)$.
Similarly $\mu(\X)$ is called the \emph{total mass} of $\mu$.
We write $\Supp{\mu}$ for the support of the measure $\mu$.

\begin{definition}
Let $\mu$ and $\nu$ be positive measures with the same total mass on a metric space $\X$.
A \emph{transport plan} between $\mu$ and $\nu$ is a measure $\pi$ on $\X\times\X$ such that for all $A,B\in\mathfrak{B}(\X)$,
\[
  \pi(A\times\X)=\mu(A)\ and\ \pi(\X\times B)=\nu(B).
\]
\end{definition}

We denote by $\Pi(\mu,\nu)$ the set of all transport plans between $\mu$ and $\nu$.
The $p$th order cost of the transport plan $\pi$ is defined as
\[
  C_p(\pi)=\left(\int_{\X\times\X}\dX{x}{y}^p d\pi(x,y)\right)^{\frac{1}{p}}.
\]
The Wasserstein distance between $\mu$ and $\nu$ is the minimum cost over all transport plans.

\begin{definition}
Let $\mu$ and $\nu$ be positive measures with the same total mass on a metric space $\X$. 
The \emph{Wasserstein distance} of order $p$ between $\mu$ and $\nu$ is defined as
$$W_p(\mu,\nu)=\min_{\pi\in\Pi(\mu,\nu)}\left(\int_{\X\times\X}\dX{x}{y}^p d\pi(x,y)\right)^{\frac{1}{p}}.$$
\end{definition}

The Wasserstein distance is finite if both probability measures have finite $p$-moments, which is always the case for measures with compact support.

\paragraph{Persistence theory\\}
  A \emph{filtration} $F = \{F_\alpha\}_{\alpha\in\R}$ is a sequence of spaces such that $F_\alpha\subseteq F_\beta$ whenever $\alpha\leq\beta$.
  Persistence theory studies the evolution of the homology of the sets $F_\alpha$ for $\alpha$ ranging from $-\infty$ to $+\infty$.
  More precisely, the filtration induces a family of vector spaces connected by linear maps at the homology level, called a \emph{persistence module}.
  More generally, a persistence module is a pair $\V=(\{V_\alpha\},\{v_\alpha^\beta\})$ where each $V_\alpha$ is a vector space and $v_\alpha^\beta$ is a linear map $V_\alpha\to V_\beta$ such that $v_\beta^\gamma\circ v_\alpha^\beta = v_\alpha^\gamma$ for all $\alpha\le\beta\le \gamma$ and $v_\alpha^\alpha$ is the identity.
  A persistence module is said to be \emph{q-tame} if $v_\alpha^\beta$ has finite rank for every $\alpha<\beta$.
  A filtration is said to be q-tame if its corresponding persistence module is q-tame.
  The algebraic structure of a q-tame persistence module $\U$ can be described and visualized by the \emph{persistence diagram} $\Dgm{\U}$, a multiset of points in the plane.
  If $\U$ comes from a filtration $\{F_\alpha\}$, a point $(\alpha,\beta)$ in $\Dgm{\U}$ indicates a nontrivial homology class that exists in the filtration between the parameter values $\alpha$ and $\beta$.
  
  We overload notation and write $\Dgm{\{F_\alpha\}}$ to denote the persistence diagram of the persistence module defined by the filtration $\{F_\alpha\}$.
  Moreover, for a real-valued function $f$, we write $\Dgm{f}$ to denote $\Dgm{\{f^{-1}(]-\infty, \alpha])\}}$, the persistence diagram of the sublevel sets filtration of $f$.
  For an introduction to persistent homology, the reader is directed to~\cite{sspmCDGO,ctaiEH}.

\paragraph{Bottleneck distance\\}
  We put a metric on the space of persistence diagrams as follows.
  First, a partial matching $M$ between diagrams $D$ and $E$ is a subset of $D\times E$ in which each element of $D\cup E$ appears in at most one pair.
  The bottleneck cost of $M$ is $\max_{(d,e)\in M}\|d-e\|_{\infty}$.
  We say $M$ is an $\epsilon$-matching if the bottleneck cost is $\epsilon$ and every $(\alpha,\beta)$ in $D$ or $E$ with $|\beta-\alpha|\ge 2\epsilon$ is matched.
  The \emph{bottleneck distance} between $D$ and $E$ is defined as
  \[
    d_B(D,E) = \inf\{\epsilon \mid \text{there exists an $\epsilon$-matching between $D$ and $E$}\}.
  \]

  It is often useful to look at persistence diagrams on a logarithmic scale, because the distance does no longer depend on the scale at which the object is seen.
  The \emph{log-bottleneck distance}, denoted $\dbl$ is the bottleneck distance between diagrams after the change of coordinates $(\alpha,\beta)\mapsto(\ln\alpha,\ln\beta)$.

\paragraph{Filtration interleaving\\}
One way to prove that two persistence diagrams are close is to prove that the filtrations inducing them are interleaved.
Two filtrations $\{U_\alpha\}_{\alpha\in\R}$ and $\{V_\alpha\}_{\alpha\in\R}$ are said to be \emph{$\epsilon$-interleaved} if for any $\alpha$,
\[
  U_\alpha\subseteq V_{\alpha+\epsilon}\subseteq U_{\alpha+2\epsilon}.
\]
The following classic result~\cite{ppmdCCGGO,sspmCDGO,spdCEH} about stability of persistence diagrams says that interleaved filtrations yield similar persistence diagrams.

\begin{theorem}\label{tstability}
Let $U$ and $V$ be two $q$-tame and $\epsilon$-interleaved filtrations. Then, the persistence diagrams of these filtrations are $\epsilon$-close in bottleneck distance, i.e.,
\[
  d_B(\Dgm{U},\Dgm{V})\leq\epsilon.
\]
\end{theorem}

We work with the persistence theory on functions, which means studying the persistence of \emph{the sublevel sets filtration} defined as $\{f^{-1}(]-\infty,\alpha])\}_{\alpha\in\R}$ for any real-valued function.
To simplify notation, we write $\Dgm f$ to denote the persistence diagram of the sublevel sets filtration of $f$.

\paragraph{Persistence module interleaving\\}
The notion of interleaving can be extended to persistence modules as seen in~\cite{psgcCDO}.
Given two persistence modules $\U=(\{U_\alpha\},\{u_\alpha^\beta\})$ and $\V=(\{V_\alpha\},\{v_\alpha^\beta\})$ and a real $\epsilon>0$, an \emph{$\epsilon$-homomorphism} from $\U$ to $\V$ is a collection of linear maps $\Phi=\{\phi_\alpha\}$ such that for all $\alpha<\beta$, $v_{\alpha+\epsilon}^{\beta+\epsilon}\circ\phi_\alpha=\phi_\beta\circ u_{\alpha}^\beta$. 
Two $\epsilon$-homomorphisms $\Phi$ from $\U$ to $\V$ and $\Psi$ from $\V$ to $\W$ can be composed to build a $2\epsilon$-homomorphism $\Psi\Phi$ from $\U$ to $\W$ whose linear maps are obtained by composing the linear maps of $\Phi$ and $\Psi$.
Among $\epsilon$-homomorphisms from $\U\to\U$, one has a particular role.
The \emph{$\epsilon$-shift map} $1_\U^{\epsilon}$ is the collection of maps $u_\alpha^{\alpha+\epsilon}$ given in the persistence module $\U$.
We use it to define the interleaving of two persistence modules as follows.

\begin{definition}
Let $\U$ and $\V$ be two q-tame persistence modules. $\U$ and $\V$ are \emph{$\epsilon$-interleaved} if there exists $\epsilon$-homomorphisms $\Phi:\U\to\V$ and $\Psi:\V\to\U$ such that $\Phi\Psi=1_\V^{2\epsilon}$ and $\Psi\Phi=1_\U^{2\epsilon}$.
\end{definition}

Note that the definition is equivalent to the commutativity of the following diagrams for any $\alpha<\beta$, where $\Phi=\{\phi_\alpha\}$ and $\Psi=\{\psi_\alpha\}$.

\begin{center}
\begin{tikzpicture}[scale=.7]
\draw (-1,1) node {$V_{\beta+\epsilon}$};
\draw (-6,1) node {$V_{\alpha+\epsilon}$};
\draw (-6,3) node {$U_\alpha$};
\draw (-1,3) node {$U_\beta$};
\draw[->] (-5.4,1) --node[below] {$v_{\alpha+\epsilon}^{\beta+\epsilon}$} (-1.6,1);
\draw[->] (-5.4,3) --node[above] {$u_\alpha^\beta$} (-1.6,3);
\draw[->] (-6,2.5) --node[left] {$\phi_\alpha^{\alpha+\epsilon}$} (-6,1.5);
\draw[->] (-1,2.5) --node[right] {$\phi_\beta^{\beta+\epsilon}$} (-1,1.5);

\draw (-1,-1) node {$V_\beta$};
\draw (-6,-1) node {$V_\alpha$};
\draw (-6,-3) node {$U_{\alpha+\epsilon}$};
\draw (-1,-3) node {$U_{\beta+\epsilon}$};
\draw[->] (-5.4,-1) --node[above] {$v_\alpha^\beta$} (-1.6,-1);
\draw[->] (-5.4,-3) --node[below] {$u_{\alpha+\epsilon}^{\beta+\epsilon}$} (-1.6,-3);
\draw[->] (-6,-1.5) --node[left] {$\psi_\alpha^{\alpha+\epsilon}$} (-6,-2.5);
\draw[->] (-1,-1.5) --node[right] {$\psi_\beta^{\beta+\epsilon}$} (-1,-2.5);

\draw (3,3) node {$U_{\alpha-\epsilon}$};
\draw (7,3) node {$U_{\alpha+\epsilon}$};
\draw (5,1) node {$V_\alpha$};
\draw[->] (3.6,3) --node[above] {$u_{\alpha-\epsilon}^{\alpha+\epsilon}$} (6.4,3);
\draw[->] (3.3,2.7) --node[left] {$\phi_{\alpha-\epsilon}^\alpha$} (4.7,1.3);
\draw[->] (5.3,1.3) --node[right] {$\psi_{\alpha}^{\alpha+\epsilon}$} (6.7,2.7);

\draw (3,-1) node {$V_{\alpha-\epsilon}$};
\draw (7,-1) node {$V_{\alpha+\epsilon}$};
\draw (5,-3) node {$U_\alpha$};
\draw[->] (3.6,-1) --node[above] {$v_{\alpha-\epsilon}^{\alpha+\epsilon}$} (6.4,-1);
\draw[->] (3.3,-1.3) --node[left] {$\psi_{\alpha-\epsilon}^\alpha$} (4.7,-2.7);
\draw[->] (5.3,-2.7) --node[right] {$\phi_\alpha^{\alpha+\epsilon}$} (6.7,-1.3);
\end{tikzpicture}
\end{center}

The following theorem is an algebraic analog of Theorem~\ref{tstability}.
The proof can be found in~\cite{sspmCDGO}.

\begin{theorem}\label{tModStability}
Let $\U$ and $\V$ be two q-tame and $\epsilon$-interleaved persistence modules. 
Then,
\[
  d_B(\Dgm{\U},\Dgm{\V})\leq\epsilon.
\]
\end{theorem}

\paragraph{Contiguous simplicial maps\\}
Let $X$ and $Y$ be simplicial complexes.
A \emph{simplicial map} $f:X\to Y$ is a map between the corresponding vertex sets so that for every simplex $\sigma\in X$, $f(\sigma) = \bigcup_{p\in \sigma}f(p)$ is a simplex in $Y$.
Two simplicial maps $f$ and $g$ are \emph{contiguous} if $\sigma \in X$ implies that $f(\sigma)\cup g(\sigma)\in Y$.
 If two simplicial maps are contiguous, then they induce the same homomorphism at the homology level~\cite[Chapter 1]{munkres84elements}.

A \emph{clique complex} is a simplicial complex whose simplices are the cliques of a graph.
Many of the simplicial complexes considered in this paper are clique complexes.
We will use the following simple lemma to construct contiguous simplicial maps between clique complexes.

\begin{lemma}\label{lem:contiguity_and_cliques}
Let $X$ and $Y$ be clique complexes and let $f$ and $g$ be two functions from the vertex set of $X$ to the vertex set of $Y$.
If for every edge $(p,q)\in X$, the tetrahedron $\{f(p), g(p), f(q), g(q)\}$ is in $Y$, then $f$ and $g$ induce contiguous simplicial maps from $X$ to $Y$.
\end{lemma}

\begin{proof}
  Let $\sigma$ be a simplex of $X$.
  Every pair in $f(\sigma)\cup g(\sigma)$ is of the form $(f(p),f(q))$, $(f(p),g(q))$, or $(g(p),g(q))$ for some vertices $p$ and $q$ in $\sigma$.
  Since $(p,q)\in \sigma$, the tetrahedron hypothesis of the lemma implies that all of these pairs are edges of $Y$.
  Thus, $f(\sigma)\cup g(\sigma)$ is a simplex in $Y$ because $Y$ is a clique complex.
  Moreover, $f(\sigma)\in Y$ and $g(\sigma)\in Y$ because simplices are closed under taking subsets.
  Therefore, $f$ and $g$ are indeed contiguous simplicial maps as desired.
\end{proof}


\section{Persistence and Stability of the Distance to a Measure in a Metric Space} 
\label{sPersStabDtm}

In this section, we prove that, if we have two close probability measures, then the persistence diagrams of the sublevel sets filtration of their distance to measure functions are close.
The result applies to \emph{triangulable} metric spaces, i.e., those that are homeomorphic to a locally finite simplicial complex.
The persistence diagrams considered in this paper are well defined in this class of spaces.
In particular, every compact Riemannian manifold is triangulable.

  If the persistence diagram is to be meaningful, one might expect that it is stable with respect to perturbations in the underlying measure.
  The following theorem shows that this is indeed the case.
  Two measures that are close in the quadratic Wasserstein distance, $W_2$ yield persistence diagrams that are close in bottleneck distance, $d_B$ (see~\cite[Sec. 7.1]{villani2003tot}).

\begin{theorem}\label{tStabPersDtm}
Let $\mu$ and $\nu$ be two probability measures on a triangulable metric space $\X$ and let $m$ be a mass parameter. 
Then $\Dgm{\dm}$ and $\Dgm{d_{\nu,m}}$ are well-defined and
$$d_B(\Dgm{\dm},\Dgm{d_{\nu,m}})\leq\frac{1}{\sqrt{m}}W_2(\mu,\nu).$$
\end{theorem}

To prove this theorem, we first show that the distance to measure functions are stable with respect to the Wasserstein distance.
Then, we prove that their diagrams are well-defined and are close using Theorem~\ref{tstability}.

\subsection{Wasserstein stability}

A measure $\nu$ is a \emph{submeasure} of a measure $\mu$ if for every $B\in\mathfrak{B}(\X),\ \nu(B)\leq\mu(B)$.
Let $\Sub{m}{\mu}$ be the set of all submeasures of $\mu$, which have a total mass $m$.

The distance to a measure $\mu$ at point a $x$ can be expressed as the Wasserstein distance between two measures,
the Dirac mass $\delta_x$ on $x$
and a submeasure of $\mu$ of mass $m$.
Using this view, we generalize the stability result from~\cite{gipmCCM} as follows. 

\begin{proposition}\label{Wwriting}
Let $\mu$ be a probability measure on a metric space $\X$, and let $m\in ]0,1]$ be a mass parameter. Then,
\[
  \dm(x)=\min_{\nu\in \Sub{m}{\mu}}\frac{1}{\sqrt{m}}W_2(m\delta_x,\nu).
\]
\end{proposition}

Given $x\in\X$ and $m>0$, let ${\cal{R}}_{\mu,m}(x)$ be the set of the submeasures of $\mu$ with total mass $m$ whose support is contained in the closed ball $\bar{B}(x,\delta_{\mu,m}(x))$ and whose restriction to the open ball $B(x,\delta_{\mu,m(x)})$ coincides with $\mu$.
The proof shows that ${\cal{R}}_{\mu,m}(x)$ is exactly the set of minimizers of Proposition~\ref{Wwriting}. 

In order to prove this theorem we need to introduce a few definitions.
The \emph{cumulative function} $F_\nu:\R^+\rightarrow\R$ of a measure $\nu$ on $\R^+$ is the non-decreasing function defined by $F_\nu(y)=\nu([0,y))$.
Its \emph{generalized inverse} $F_\nu^{-1}:m\mapsto\inf\{t\in\R\mid F_\nu(t)>m\}$ is left-continuous. 

\begin{proof}
If $\nu$ is a measure of total mass $m$ on $\X$ then there exists only one transport plan between $\nu$ and the Dirac mass $m\delta_x$.
It transports every point of $\X$ to $x$.
Hence we get
\[
  W_2(m\delta_x,\nu)^2=\int_{\X}\dX{h}{x}^2\ d\nu(h).
\]

Let $d_x:\X\rightarrow\R$ denote the distance function to the point $x$ and let $\nu_x$ be the pushforward of $\nu$ by the distance function to $x$.
That is, for any subset $I$ of $\R, \nu_x(I)=\nu(d_x^{-1}(I))$. 
Note that $F^{-1}_{\nu_x}(m)=\delta_{\nu,m}(x)$.
Using the change of variable formula and the definition of the cumulative function we get:
$$\int_{\X}\dX{h}{x}^2d\nu(h)=\int_{\R^+} t^2d\nu_x(t)=\int_0^mF_{\nu_x}^{-1}(l)^2dl.$$

Suppose further that $\nu$ is a submeasure of $\mu$, then $F_{\nu_x}(t)\leq F_{\mu_x}(t)$ for all $t>0$. 
So, $F_{\nu_x}^{-1}(l)\geq F_{\mu_x}^{-1}(l)$ for all $l> 0$, and thus,
\begin{equation}\label{cumulineq}
  W_2(m\delta_x,\nu)^2\geq\int_0^mF_{\mu_x}^{-1}(l)^2dl=\int_0^m\delta_{\mu,l}(x)^2dl=m\dm(x)^2.
\end{equation}
This inequality implies that $\dm(x)$ is smaller than $\frac{1}{\sqrt{m}}W_2(m\delta_x,\nu)$ for any $\nu\in\Sub{m}{\mu}$.

Consider the case when the inequality in (\ref{cumulineq}) is tight.
Such a case happens when for almost every $l\leq m,\ F_{\nu_x}^{-1}(l)=F_{\mu_x}^{-1}(l)$.
Since these functions are increasing and left-continuous, equality must hold for every such $l$.
By the definition of the pushforward, this implies that $\nu(\bar{B}(x,\delta_{\mu,m}(x)))=m$, i.e., all the mass of $\nu$ is contained in the closed ball $\bar{B}(x,\delta_{\mu,m}(x))$, and that $\nu(B(x,\delta_{x,\mu}(m)))=\mu(B(x,\delta_{x,\mu}(m)))$.
Because $\nu$ is a submeasure of $\mu$ this is true if and only if $\nu$ is in the set ${\cal R}_{\mu,m}(x)$ described before the proof.
Thus ${\cal R}_{\mu,m}(x)$ is exactly the set of submeasures $\nu\in\Sub{m}{\mu}$ such that $\dm(x)=\frac{1}{\sqrt{m}}W_2(m\delta_x,\nu)$.

To conclude the proof we need only show that there exists at least one measure $\mu_{x,m}$ in the set ${\cal R}_{\mu,m}(x)$.
If $\mu(\bar{B}(x,\delta_{\mu,m}(x)))=m$, then $\mu_{x,m}=\mu|_{\bar{B}(x,\delta_{\mu,m}(x))}$ is an obvious choice.
The only difficulty is when the boundary $\partial B(x,\delta_{\mu,m}(x))$ of the ball has too much mass.
In this case we uniformly rescale the mass contained in the bounding sphere such that the measure $\mu_{x,m}$ has total mass $m$. 
More precisely we let:
$$\mu_{x,m}=\mu|_{B(x,\delta_{\mu,m}(x))}+(m-\mu(B(x,\delta_{\mu,m}(x))))\frac{\mu|_{\partial B(x,\delta_{\mu,m}(x))}}{\mu(\partial B(x,\delta_{\mu,m}(x)))}.$$
We hence have $\frac{1}{\sqrt{m}}W_2(m\delta_x,\mu_{x,m})=\dm(x)$.
\end{proof}

From this result, we have the following Wasserstein stability guarantee for the distance to a measure.

\begin{theorem}\label{stability}
Let $\mu$ and $\nu$ be two probability measures on a metric space $\X$ and let $m\in]0,1]$ be a mass parameter. 
Then:
$$\norm{\dm-d_{\nu,m}}_\infty\leq\frac{1}{\sqrt{m}}W_2(\mu,\nu).$$
\end{theorem}

\begin{proof}
Using Proposition~\ref{Wwriting}, we get that $\sqrt{m}\ \dm(x)=W_2(m\delta_x,\mu_{x,m})$, where $\mu_{x,m}\in {\cal R}_{\mu,m}(x)$. 
Let $\pi$ be an optimal transport plan between $\mu$ and $\nu$, i.e., a transport plan between $\mu$ and $\nu$ such that
\[
  \int_{\X\times\X}\dX{x}{y}^2d\pi(x,y)=W_2(\mu,\nu)^2.
\]

Let us consider the submeasure $\mu_{x,m}$ of $\mu$.
Then there exists $\tilde\pi$ a submeasure of $\pi$ that transports $\mu_{x,m}$ to a submeasure $\tilde\nu$ of $\nu$.
We get that:
$$W_2(\mu_{x,m},\tilde\nu)\leq W_2(\mu,\nu).$$
Using Proposition~\ref{Wwriting} again, we get that for any $x\in\X$, $\sqrt{m}\ d_{\nu,m}(x)\leq W_2(m\delta_x,\tilde\nu)$.
Thus,
\begin{align*}
\sqrt{m}\ d_{\nu,m}(x)
&\leq W_2(m\delta_x,\tilde\nu)\leq W_2(m\delta_x,\mu_{x,m})+W_2(\tilde\nu,\mu_{x,m})\\
&\leq \sqrt{m}\ \dm(x)+W_2(\mu,\nu).
\end{align*}
The roles of $\mu$ and $\nu$ can be reversed to conclude the proof.
\end{proof}

Another consequence of Proposition~\ref{Wwriting} is that $\dm$ is $1$-Lipschitz with respect to $x$.

\begin{proposition}\label{pLipschitz}
Let $\mu$ be a probability measure on a metric space $\X$ and let $m\in]0,1]$ be a mass parameter. Then $\dm$ is 1-Lipschitz.
\end{proposition}

\begin{proof}
Let $x$ and $y$ be two points of $\X$. 
Using Proposition~\ref{Wwriting}, there exists a submeasure $\mu_{x,m}$ of $\mu$ such that $\dm(x)=\frac{1}{\sqrt{m}}W_2(m\delta_x,\mu_{x,m})$.
The same proposition applied to $y$ gives $\dm(y)\leq \frac{1}{\sqrt{m}}W_2(m\delta_y,\mu_{x,m})$.
Knowing that $W_2(m\delta_x,m\delta_y)=\sqrt{m}\ \dX{x}{y}$, we can conclude that $\dm(y)\leq\dm(x)+\dX{x}{y}$.
The choice of $x$ and $y$ is arbitrary, so by symmetry, $\dm(x)\leq\dm(y)+\dX{x}{y}$.
Therefore, $\dm$ is 1-Lipschitz.
\end{proof}

\subsection{Persistence}
For persistence diagrams of sublevel sets filtrations of distance to measure functions to be well-defined, we need to prove that they are q-tame.
\begin{proposition}\label{pQtame}
Let $\X$ be a triangulable metric space, let $\mu$ be a probability measure on $\X$, and let $m\in]0,1]$ be a mass parameter. 
Then, the sublevel sets filtration of $\dm$ is $q$-tame.
\end{proposition}


\begin{proof}
According to Proposition~\ref{pLipschitz} $\dm$ is 1-Lipschitz and thus continuous. 
Also, $\dm$ is nonnegative by definition.
Moreover, $\dm$ is proper, i.e., the preimage of any compact set is compact. 
As the function is nonnegative and continuous, it suffices to show that any sublevel set $\dm^{-1}([0,\alpha])$ is compact.

Suppose for contradiction that for a fixed $\alpha>0$, $\dm^{-1}([0,\alpha])$ is not compact.
Then there exists a sequence $(x_i)_{i>0}$ of points of $\dm^{-1}([0,\alpha])$ such that $\dX{x_0}{x_n}\rightarrow\infty$ when $n\rightarrow\infty$.
Hence we can extract a sub-sequence $(x_{\phi(i)})_{i>0}$ such that for any $i$ and $j$, $\bar{B}(x_{\phi(i)},\sqrt{2}\alpha)\cap \bar{B}(x_{\phi(j)},\sqrt{2}\alpha)=\emptyset$.
Let us remark that $\mu(\bar{B}(x_{\phi(i)},\sqrt{2}\alpha))\geq\frac{m}{2}$.
So,
\[
  \dm(x_{\phi(i)})^2=\frac{1}{m}\int_0^m\delta_{\mu,l}(x_{\phi(i)})^2 dl\leq\alpha^2.
\]
The function $\delta_{\mu,l}(x_{\phi(i)})$ is nonnegative and increasing with $l$ and therefore $\frac{m}{2}\delta_{\mu,\frac{m}{2}}(x_{\phi(i)})^2\leq m\alpha^2$.
Using the definition of $\delta_{\mu,m}$, this implies that $\mu(\bar{B}(x_{\phi(i)},\sqrt{2}\alpha))\geq\frac{m}{2}$.
Measures are countably additive, so 
\[
  \mu(\X)\geq\sum_{i>0}\mu(\bar{B}(x_{\phi(i)},\sqrt{2}\alpha))\geq\sum_{i>0}\frac{m}{2}=\infty.
\]
However, $\mu$ is a probability measure and therefore $\mu(\X)=1$.
This contradiction implies that $\dm^{-1}([0,\alpha])$ is compact.

As $\X$ is triangulable, there exists a homeomorphism $h$ from $\X$ to a locally finite simplicial complex $C$.
Then for any $\alpha>0$, we can restrict the simplicial complex $C$ to a finite simplicial complex $C_\alpha$ that contains $h(\dm^{-1}([0,\alpha]))$ as $\dm^{-1}([0,\alpha])$ is compact.
The function $\dm\circ h^{-1}|_{C_\alpha}$ is continuous on $C_\alpha$.
Thus its sublevel sets filtration is $q$-tame by Theorem~2.22 of~\cite{sspmCDGO}.

The construction extends to any $\alpha$ and therefore the sublevel sets filtration of $\dm\circ h^{-1}$ is $q$-tame.
Furthermore, homology is preserved by homeomorphisms and thus we can say that the sublevel sets filtration of $\dm$ is $q$-tame.
\end{proof}

Theorem~\ref{tStabPersDtm} is now obtained by combining Theorem~\ref{tstability} and Proposition~\ref{pQtame}.

\begin{proof}[Proof of Theorem~\ref{tStabPersDtm}]
Theorem~\ref{stability} guarantees that:
$$\norm{\dm-d_{\nu,m}}_\infty\leq\frac{1}{\sqrt{m}}W_2(\mu,\nu).$$
The sublevel sets filtrations are therefore interleaved since for all $\alpha\in\R$,
\[
  \dm^{-1}(]-\infty,\alpha])
    \subseteq 
  d_{\nu,m}^{-1}(]-\infty,\alpha+\frac{1}{\sqrt{m}}W_2(\mu,\nu)])
    \subseteq
  \dm^{-1}(]-\infty,\alpha+\frac{2}{\sqrt{m}}W_2(\mu,\nu)]).
\]
Therefore, applying Theorem~\ref{tstability} gives
\[
  d_B(Dgm(\dm),Dgm(d_{\nu,m}))\leq\frac{1}{\sqrt{m}}W_2(\mu,\nu).
\]
\end{proof}


\section{Approximating the Distance to a Measure} 
\label{sec:approximating_the_distance_to_a_measure}

Computing the persistence diagram of the sublevel sets filtration of $\dm$ requires knowing the sublevel sets. 
They are not generally easy to compute. 
We propose an approximation paradigm for $\dm$ that replaces the sublevel sets by a union of balls.
The approach works in any metric space and yields equivalent guarantees as the witnessed $k$-distance approach used in~\cite{wkdGMM} for Euclidean spaces.

\subsection{Power distances}\label{sApd}


\begin{definition}
  Given a metric space $\X$, a set $P$ and a function $w:P\to\R$, we define the \emph{power distance} $f$ associated with $(P,w)$ as
  \begin{equation}\label{eq:power_distance}
    f(x)=\sqrt{\min_{p\in P}\dX{p}{x}^2+w_p^2},
  \end{equation}
  where $w_p$ is the value of $w$ at the point $p$.
\end{definition}

The function $w$ can be defined on a superset of $P$.
Moreover, the sublevel set $f^{-1}(]-\infty,\alpha])$ is the union of the closed balls centered on the points $p$ of $P$ with radius $r_p(\alpha)=\sqrt{\alpha^2-w_p^2}$.
By convention, we assume the ball is empty when the radius is imaginary.

\paragraph{Stability\\}
Power distances are stable under small perturbations of the points.

%

The following lemma states a result about inclusions between balls.
It allows another stability result on power distances (Proposition~\ref{pPowerBalls}) and will be useful for studying the stability of the weighted Rips filtration in Section~\ref{sWRips}.

  \begin{lemma}\label{lPowerBalls}
    Let $p,q\in\X$ be such points such that $\dX{p}{q}\le \epsilon$, and let $w:\X\to\R$ be a $t$-Lipschitz function.
    For all $\alpha \ge w_p$,
    \[
      r_p(\alpha) + \epsilon \le r_q(\alpha+\sqrt{1+t^2}\;\epsilon).
    \]
  \end{lemma}
  \begin{proof}
    First, observe that $r_p(\alpha)$ can be bounded as follows.
    \begin{align*}
      r_p(\alpha)^2
        = \alpha^2 - w_p^2
        &\le \alpha^2 - w_p^2 + (t\alpha - \sqrt{1+t^2}\;w_p)^2\\
        &= (\sqrt{1+t^2}\;\alpha - t w_p)^2.
    \end{align*}
    Next, we relate $r_p$ and $r_q$ as follows.
    \begin{align*}
      (r_p(\alpha)+\epsilon)^2
        & =   \alpha^2 - w_p^2 + 2 \epsilon \sqrt{\alpha^2 - w_p^2} + \epsilon^2\\
        & \le \alpha^2 - w_p^2 + 2 \epsilon (\sqrt{1+t^2}\;\alpha - t w_p) + \epsilon^2\\
        & =   (\alpha + \sqrt{1+t^2}\;\epsilon)^2 - (w_p + t \epsilon)^2\\
        & \le (\alpha+\sqrt{1+t^2}\;\epsilon)^2-w_q^2\\
        & =   r_q(\alpha+\sqrt{1+t^2}\;\epsilon)^2.
    \end{align*}
    The requirement that $\alpha \ge w_p$ allows us to take the square root of both sides of the inequality since both will be nonnegative.
  \end{proof}
  \noindent
  As a consequence, we obtain the following.

  \begin{proposition}\label{pPowerBalls}
    Let $\X$ be a metric space and let $w:\X\to\R$ be a function.
    Let $P$ and $Q$ be two compact subsets of $\X$.
    Let $f_P$ and $f_Q$ be the power distances associated with $(P,w)$ and $(Q,w)$.
    If $w$ is $t$-Lipschitz, then
    \[
      \norm{f_P-f_Q}_\infty \le \sqrt{1+t^2}\;d_H(P,Q).
    \]
  \end{proposition}
  \begin{proof}
    Let $x$ be any point of $\X$.
    There exists $p\in P$ such that $x\in \bar{B}(p,r_p(f_P(x)))$.
    There also exists $q\in Q$ such that $\dX{p}{q} \le d_H(P,Q)$.
    By Lemma~\ref{lPowerBalls} and the triangle inequality, $x\in\bar{B}(q,r_q(f_P(x) + \sqrt{1+t^2}\;d_H(P,Q)))$.
    Thus, $f_Q(x) \le f_P(x) + \sqrt{1+t^2}\;d_H(P,Q)$.
    $P$ and $Q$ are interchangeable therefore $\norm{f_Q-f_P}_\infty \le \sqrt{1+t^2}\;d_H(P,Q)$.
    %
  \end{proof}

\paragraph{Approximation\\}
To approximate the distance to a probability measure $\mu$, we introduce the following function.

\begin{definition}\label{dDpd}
  Let $\mu$ be a probability measure on a metric space $\X$ and let $m\in]0,1]$ be a mass parameter.
  Given a subset $P$ of $\X$, we define $\dPm$ as the power distance associated with $(P,\dm)$.
  \[
    \dPm(x)=\sqrt{\min_{p\in P}\dm(p)^2+\dX{p}{x}^2}  
  \]
\end{definition}
That is, the weight of each point is its distance to the empirical measure.
If $P$ is close to $\Supp{\mu}$, we obtain an approximation of $\dm$.

\begin{theorem}\label{tPbound}
Let $\mu$ be a probability measure on a metric space $\X$ and let $m\in]0,1]$ be a mass parameter.
Let $P$ be a subset of $\X$.
If $P$ is an $\epsilon$-sample of $\Supp{\mu}$, then
$$\frac{1}{\sqrt{2}}\dm\leq\ \dPm\leq\sqrt{5}\ (\dm+\epsilon).$$
\end{theorem}

A multiplicative approximation implies a multiplicative interleaving of the sublevel sets filtrations that becomes an additive interleaving on a logarithmic scale.
Theorem~\ref{tstability} thus guarantees that the persistence diagrams are close in the bottleneck distance on a logarithmic scale.

\begin{proof}
Let $x$ be a point of $\X$. 
Using the previous notations we get
\[
  \dm(x)^2=\frac{1}{m}\int_{\X}\dX{y}{x}^2\mu_{x,m}(y)\dy.
\]
Let us now fix a point $p\in \Supp{\mu}$.
Since $\mu_{p,m}$ is a submeasure of $\mu$ of total mass $m$,
\begin{align*}
\dm(x)^2
&=\frac{1}{m}\int_{\X}\dX{y}{x}^2\mu_{x,m}(y)\dy\\
&\leq\frac{1}{m}\int_{\X}\dX{y}{x}^2\mu_{p,m}(y)\dy\\
&\leq\frac{1}{m}\int_{\X}((\dX{y}{p}+\dX{p}{x})^2)\mu_{p,m}(y)\dy\\
&\leq\dX{p}{x}^2\ \frac{2}{m}\int_\X\mu_{p,m}(y)\dy+\frac{2}{m}\int_{\X}\dX{y}{p}^2\mu_{p,m}(y)\dy\\
&=2(\dX{p}{x}^2+\dm(p)^2).
\end{align*}
The third inequality follows from the triangle inequality and the relation $(a+b)^2\leq2(a^2+b^2)$.

As the above inequality holds for any point $p$ in $P$ we can conclude that
\[
  \dm(x)\leq\sqrt{2}\ \dPm(x).
\]

To show the other inequality, let $p$ be a point of $P$.
Then by definition we get:
\begin{align*}
\dPm(x)^2
&\leq\dX{x}{p}^2+\dm(p)^2\\
&\leq\dX{x}{p}^2+\frac{1}{m}\int_{\X}\dX{p}{y}^2\mu_{x,m}(y)\dy\\
&\leq\dX{x}{p}^2+\frac{1}{m}\int_{\X}(\dX{p}{x}+\dX{x}{y})^2\mu_{x,m}(y)\dy\\
&\leq3\ \dX{x}{p}^2+2\ \dm(x)^2.
\end{align*}
By the definition of the distance to a measure, $\dX{x}{\Supp{\mu}}\leq\dm(x)$.
Consequently, there exists a point $p\in P$ such that $\dX{x}{p}\leq\dm(x)+\epsilon$.
Hence,
\[
  \dPm(x)^2\leq5(\dm(x)+\epsilon)^2.\qedhere
\]
\end{proof}

\subsection{Measures with finite support}\label{sRestrict}

We now assume that the data are given as a finite set of points $P$ in a metric space $\X$.
We define the following measure to study the point set $P$.

\begin{definition}
Given a finite point set $P$ in a metric space $\X$, the \emph{empirical measure} $\mu_P$ on $P$ is defined as a normalized sum of Dirac measures:
\[
  \mu_P=\frac{1}{|P|}\sum_{p\in P}\delta_p.
\]
\end{definition}

Let $x$ be a point of $\X$. 
We introduce the parameter $k=m|P|$.
To simplify the exposition we will assume that $k$ is an integer.
See Remark~\ref{rkNotInteger} for the generalization.

We reorder the points of $p$ such that $P=(p_1(x),\cdots,p_{|P|}(x))$ and 
\begin{equation}\label{eOrdering}
  \dX{x}{p_1(x)}\leq\cdots\leq\dX{x}{p_{|P|}(x)}.
\end{equation}
If two points are at the same distance of $x$, we order them arbitrarily.
We define the set 
\[
  NN_k^P(x)=\{p_1(x),\cdots,p_k(x)\}
\]
and call it the set of $k^{th}$ nearest neighbors of $x$.
The set $\Lambda_k^P$ consists of all $k$-tuples of points of $P$.

\begin{lemma}\label{Ewriting}
Let $P$ be a finite point set in a metric space $\X$ then for any $x\in\X$:
$$\dmP(x)=\sqrt{\min_{S\in\Lambda_k^P}\frac{1}{k}\sum_{p\in S}\dX{p}{x}^2}=\sqrt{\frac{1}{k}\sum_{p\in NN_k^P(x)}\dX{p}{x}^2}.$$
\end{lemma}

\begin{proof}
Since $\mu_P$ has finite support, all its submeasures also have finite support.
\[
  \Sub{m}{\mu_P}=\left\{\sum_{p\in P}\lambda_p\delta_p\mid \forall p\in P,\ 0\leq\lambda_p\leq \frac{1}{|P|} \ \mathrm{and}\ \sum_{p\in P}\lambda_p=m\right\}
\]
Let $\nu=\sum_{p\in P}\lambda_p\delta_p$ be an element of $\Sub{m}{\mu_P}$.
$$W_2(m\delta_x,\nu)^2=\sum_{p\in P}\lambda_p\dX{x}{p}^2$$
Combined with the relation~(\ref{eOrdering}), we get
\[
  S_x=\sum_{p\in NN_k^P(x)}\delta_p\in\argmin_{\nu\in\Sub{m}{\mu_P}}W_2(m\delta_x,\nu).
\]
As $S_x\in\Lambda_k^P$, we are done.
\end{proof}

The distance to the empirical measure, $\dmP$, is thus defined as a lower envelope of quadratic functions.
It is generally costly if not impossible to compute its sublevel sets.

However, we can directly use the approximation presented in Section~\ref{sApd}.
Using $P$ in Definition~\ref{dDpd} and Theorem~\ref{tPbound}, we get the following.

\begin{corollary}\label{cPbound}
Let $P$ be a finite point set of a metric space $\X$ and $m\in]0,1]$ be a mass parameter.
Then,
$$\frac{1}{\sqrt{2}}\dmP\leq\ \dPP\leq\sqrt{5}\ \dmP.$$
\end{corollary}

The multiplicative approximation gives a closeness result between persistence diagrams on a logarithmic scale.

\begin{corollary}\label{cLogbound}
Let $P$ be a finite point set of a triangulable metric space $\X$ and $m\in]0,1]$ be a mass parameter.
Then,
$$\dbl(\Dgm{\dmP},\Dgm{\dPP})\leq\ln(\sqrt{5}).$$
\end{corollary}

\begin{proof}
Corollary~\ref{cPbound} implies that
$$\ln(\dmP)-\ln(\sqrt{2})\leq\ln(\dPP)\leq\ln(\sqrt{5})+\ln(\dmP).$$
The sublevel sets of $\ln(\dmP)$ and $\ln(\dPP)$ are thus $\ln(\sqrt{5})$-interleaved and Theorem~\ref{tstability} applies.
\end{proof}

Moreover, these bounds cannot be improve.

\begin{proposition}
The bounds of Corollary~\ref{cPbound} are tight.
\end{proposition}

\begin{proof}
We are looking for a worst case scenario where inequalities become equalities for at least one point.  
We consider the space $\R^d$ with the $L_1$-norm, denoted $|\cdot|$.
For any fixed dimension $d$, we build the set of $2d$ points whose coordinates have the form $(0, \cdots, 0, \pm 1, 0, \cdots, 0)$. 
These points are marked by triangles in the following drawing in dimension $2$.
\begin{center}
\begin{tikzpicture}
\draw[thin] (-4,0) -- (4,0);
\draw[thin] (0,-2) -- (0,2);
\draw[red,thick] (-1.1,-.1) -- (-.9,-.1) -- (-1,.1) -- (-1.1,-.1);
\draw[red,thick] (1.1,-.1) -- (.9,-.1) -- (1,.1) -- (1.1,-.1);
\draw[red,thick] (-.1,-1.1) -- (.1,-1.1) -- (0,-.9) -- (-.1,-1.1);
\draw[red,thick] (0,1.1) -- (-.1,.9) -- (.1,.9) -- (0,1.1);
\draw[blue,thick] (-3,0) circle (3pt);
\draw[blue,thick] (0,0) circle (3pt);
\draw (-3,.3) node {$q$};
\draw (-1,.3) node {$p_1$};
\draw (1,.3) node {$p_2$};
\draw (.3,1.3) node {$p_3$};
\draw (.3,-.7) node {$p_4$};
\draw (.3,.3) node {$o$};
\end{tikzpicture}
\end{center}
We fix $k=2d$ and we study $\dmP$ and $\dPP$ at points $q(-3,0\cdots,0)$ and $o$.
First we compute the value of $\dmP(p_i)$ for any $i$:
$$\dmP(p_i)^2=\frac{1}{2d}\sum_{j=1}^{2d}|p_j-p_i|^2=4\frac{2d-1}{2d}=4-\frac{2}{d}$$
Now we compute the value of $\dmP$ at $q$ and $o$:
$$\dmP(o)^2=\frac{1}{2d}\sum_{i=1}^{2d}|p_i-o|^2=1$$
$$\dmP(q)^2=\frac{1}{2d}\sum_{i=1}^{2d}|p_i-q|^2=\frac{1}{2d}(|p_1-q|^2+(2d-1)16)=16-\frac{6}{d}$$
All the points $p_i$ have the same value for $\dmP$.
It is easy to compute $\dPP$ at $q$ and $o$
$$\dPP(o)^2=\dmP(p_1)^2+|p_1-o|^2=5-\frac{2}{d}$$
$$\dPP(q)^2=\dmP(p_1)^2+|p_1-q|^2=8-\frac{2}{d}$$ When $d$ increases,
the ratio $\frac{\dPP(o)}{\dmP(o)}$ tends to $\sqrt{5}$, while
$\frac{\dPP(q)}{\dmP(q)}$ tends to $\frac{1}{\sqrt{2}}$.  
Thus, the bounds of Corollary~\ref{cPbound} are reached at the limit for a same data set, although at two different points.
\end{proof}

\begin{remark}\label{rkNotInteger}
If $k$ is not an integer, it suffices to do the same construction with a careful weighting of the point $p_{\lceil k\rceil}$.
The results stay exactly the same after replacing $k$ by $\lceil k\rceil$.
\end{remark}

\subsection{Euclidean case}\label{ssEc}

We consider the standard Euclidean space $\R^d$ with the $L_2$-norm.
Considering the finite point set $P$ and its empirical measure in $\R$, we are able to express the distance to the empirical measure $\dmP$ as a power distance. 
This restricted settings allows us to improve the bounds of Corollary~\ref{cPbound} as follows.
\begin{theorem}\label{tPboundEuc}
Let $P$ be a finite point set in $\R^d$ and let $m\in]0,1]$ be a mass parameter.
Then the following relation is tight.
\[
  \frac{1}{\sqrt{2}}\dmP\leq\dPP\leq\sqrt{3}\ \dmP.
\]
Moreover, it implies a relation between persistence diagrams:
\[
  \dbl(\Dgm{\dmP},\Dgm{\dPP})\leq\ln(\sqrt{3}).
\]
\end{theorem}

We first present a way to express the distance to a measure as a power distance to the set of all barycenters of $k$-tuples of $P$.
Then we prove Theorem~\ref{tPboundEuc} before comparing it with the previous approximation, called the witnessed $k$-distance proposed in~\cite{wkdGMM}.
We improve the bounds on the witnessed $k$-distance and show that the quality of the approximation is the same for both functions.

\subsubsection{Power distance expression of $\dmP$}
For a fixed integer $k$, the barycenter associated with a point $x$ is the barycenter of its $k$-nearest neighbors.
It is also the center of the cell of the $k^{th}$-order Voronoi diagram that contains $x$.

\begin{definition}
For a point set $P$ in $\R^d$ and an integer $k\leq|P|$, the \emph{barycenter associated with $x$} is
\[
  \bary{x}=\frac{1}{k}\sum_{p\in NN_k^P(x)} p.
\]
\end{definition}

Any subset of $k$ elements from $P$ is uniquely associated with a barycenter. 
We identify the two objects and define a cell energy that describes how clustered the points are.

\begin{definition}
Let $P$ be a point set of $\R^d$ and let $k\leq|P|$.
Given $S\in\Lambda_k^P$, we fix $q=\frac{1}{k}\sum_{p\in S}p$ and define the cell energy as
\[
  E^C(q)=\frac{1}{k}\sum_{p\in S}\norm{p-q}^2.
\]
\end{definition}

Notice that the set $S$ is not necessarily the set $NN_k^P(q)$ and that $E^C(q)\geq\dmP(q)^2$.
We can now write $\dmP$ in the following form.

\begin{lemma}\label{lempiricaldef}
Let $P$ be a finite point set of $\R^d$ let $m\in]0,1]$ be a mass parameter.
For any $x\in\R^d$,
$$\dmP(x)=\sqrt{\min_{y\in \R^d}E^C(\bary{y})+\norm{\bary{y}-x}^2}=\sqrt{E^C(\bary{x})+\norm{\bary{x}-x}^2}.$$
\end{lemma}

\begin{proof}
Fix $S\in\Lambda_k^P$ and write $q=\frac{1}{k}\sum_{p\in S}p$. 
We adapt Lemma~\ref{Ewriting} to the Euclidean setting to get
\[
  \frac{1}{k}\sum_{p\in S}\norm{p-x}^2=E^C(q)+\norm{q-x}^2.
\]
This requires the inner product as follows.
\begin{align*}
A=\frac{1}{k}\sum_{p\in S}\norm{p-x}^2&=\frac{1}{k}\sum_{p\in S}\left(\norm{p-q}^2+\norm{q-x}^2+2\langle p-q|q-x\rangle\right)\\
&=E^C(q)+\norm{q-x}^2+2\langle q-q|q-x\rangle.
\end{align*}

Lemma~\ref{Ewriting} guarantees that
\[
  \dmP(x)=\sqrt{\min_{S\in\Lambda_k^P}\frac{1}{k}\sum_{p\in S}\norm{p-x}^2}=\sqrt{\frac{1}{k}\sum_{p\in NN_k^P(x)}\norm{p-x}^2},
\]
and thus,
\[
  \dmP(x) 
  = \sqrt{\min_{S\in\Lambda_k^P}E^C(q)+\norm{q-x}} 
  = \sqrt{E^C(\bary{x})+\norm{x-\bary{x}}^2}.
\]
\end{proof}

In Euclidean space, it is possible to compute the sublevel sets of $\dmP$ exactly.
The function is a power distance and its sublevel sets are unions of balls.
However, the complexity problem pointed out in section~\ref{sRestrict} is still valid. 
The number of balls required to describe a sublevel set is $\Omega(k^{\lceil\frac{d+1}{2}\rceil}n^{\lfloor\frac{d+1}{2}\rfloor})$~\cite{arscgCS}.

\subsubsection{Proof of Theorem~\ref{tPboundEuc}}

\begin{proof}
The first inequality is exactly the same as the one from Theorem~\ref{tPbound}.
For the second inequality, let $x$ be a point in $\R^d$, and let $p$ be a point of $P$.
Thus,
\[
  \dPP(x)^2\leq \dmP(p)^2+\norm{p-x}^2.
\]

Using Lemma~\ref{lempiricaldef}, we get, 
\[
  \dPP(x)^2\leq E^C(\bary{x})+\norm{p-\bary{x}}^2+\norm{p-x}^2,
\]
and with the inner product, this becomes
\begin{align*}
  \dPP(x)^2
    &\leq E^C(\bary{x})+\norm{x-\bary{x}}^2+2\norm{p-x}^2+2<x-\bary{x}|p-x>\\
    &=\dmP(x)^2+2\norm{p-x}^2+2<x-\bary{x}|p-x>.
\end{align*}
Note that 
\[
  2<\bary{x}-x|x-p>=\norm{\bary{x}-p}^2-\norm{\bary{x}-x}^2-\norm{x-p}^2.
\]
Then we can write the following relation. 
\[
  \dPP(x)^2\leq \dmP(x)^2+\norm{p-x}^2+\norm{\bary{x}-p}^2-\norm{x-\bary{x}}^2.
\]

This relation holds for any point of $P$. 
In particular it holds for any of the $k$ nearest neighbors of $x$. 
If we take the average over the $k$ nearest neighbors of $x$ and eliminate the negative term $-\norm{x-\bary{x}}^2$, we obtain
\[
  \dPP(x)^2 \leq \dmP(x)^2+\frac{1}{k}\sum_{p\in NN_k^P(x)} \norm{p-x}^2+ \frac{1}{k}\sum_{p\in NN_k^P(x)}\norm{\bary{x}-p}^2.
\]

Using the definitions of the cell energy and of the distance to the measure, we can write: $$\dPP(x)^2\leq \dmP(x)^2+\dmP(x)^2+E^C(\bary{x})$$
where $E^C(\bary{x})\leq\dmP(x)^2$.
We conclude that 
\[
  \dPP(x)\leq \sqrt{3}\ \dmP(x).
\]

The relation between persistence diagrams is follows exactly as in the proof of Corollary~\ref{cLogbound}.

\paragraph{Tightness\\}
The tight example is the point set $P$ of two points $a$ and $b$ on the real line with coordinates $1$ and $-1$.
\begin{center}
\begin{tikzpicture}
\draw (-3,0) -- (3,0);
\draw[red,thick] (-1.1,0)--(-.9,0) (-1,.1)--(-1,-.1);
\draw[red,thick] (1.1,0)--(.9,0) (1,.1)--(1,-.1);
\draw[blue,thick] (-.1,0)--(.1,0) (0,.1)--(0,-.1);
\draw (-1,-.4) node {-1};
\draw (0,-.4) node {0};
\draw (1,-.4) node {1};
\draw (-1,.4) node {$a$};
\draw (1,.4) node {$b$};
\draw (0,.4) node {$o$};
\end{tikzpicture}
\end{center}

Fix the mass parameter $m$ equal to $1$ so that $k=2$. 
It follows that
\[
  \dmP(a)=\dmP(b)=\sqrt{\frac{1}{2}\norm{b-a}^2}=\sqrt{2},
\]
and
\[
  \dmP(o)=\sqrt{\frac{1}{2}\norm{o-b}^2+\norm{o-a}^2}=1.
\]

We now compute the last interesting value:
$$\dPP(o)^2=\dmP(a)^2+\norm{a-o}^2=3.$$
We can thus conclude that $\dPP(o)=\sqrt{3}\ \dmP(o)$.
\end{proof}

\subsubsection{Comparison with witnessed $k$-distance}\label{ssCompWkd}

Another way of approximating $\dmP$ was proposed in~\cite{wkdGMM}.
Taking advantage of the power distance expression of $\dmP$, it reduced the set of barycenters to consider.
Selecting only the barycenter which are associated with the $k$ nearest neighbors of a point of $P$ gives a set of size at most $|P|$.

\begin{definition}
Let $P$ be a finite point set of $\R^d$ and let $m\in]0,1]$ be a mass parameter.
The \emph{witnessed $k$-distance} is defined as
$$\dPW(x)=\sqrt{\min_{p\in P}E^C(\bary{p})+\norm{\bary{p}-x}^2}.$$
\end{definition}

A bound on the quality of the approximation was given in Lemma 3.3 of~\cite{wkdGMM}.
We improve this bound and prove it to be at least as good as our approximation. 
We are not able to prove the tightness of this bound.
However, we can give a lower bound on the precision.
Using $\dPP$ will not improve the results compared to the witnessed $k$-distance but will not downgrade the quality either.
Moreover it can be used in a more general setting as we do not need the existence of the barycenters.

\begin{theorem}\label{tWbound}
Let $P$ be a finite point set of $\R^d$ and let $m\in]0,1]$ be a mass parameter.
Then,
$$\dmP\leq\dPW\leq\sqrt{6}\ \dmP.$$
\end{theorem}

The previous version of this theorem used a $3$ instead of the $\sqrt{6}$.

\begin{proof}
The first inequality is obtained by noticing that $\dPW$ is a minimum over a smaller set than $\dmP$.
We thus get $\dmP\leq\dPW$.

Let $x$ be a point in $\mathbb{R}^d$.
Thus for any $p\in P$,
\begin{align*}
\dm^W(x)^2&\leq E^C(\bary{p})+\norm{\bary{p}-x}^2\\
&\leq E^C(\bary{p})+\norm{\bary{p}-p}^2+\norm{p-x}^2+2\langle \bary{p}-p|p-x\rangle\\
&\leq \dm(p)^2+2\norm{p-x}^2+\norm{\bary{p}-p}^2\\
&\leq 2(\dm(p)^2+\norm{p-x}^2)\\
&\leq 2\ \dm^P(x)^2.
\end{align*}

Hence using Theorem~\ref{tPboundEuc} we can conclude that:
$$\dm^W(x)\leq\sqrt{2}\ \dm^P(x)\leq\sqrt{6}\ \dm(x).$$
\end{proof}

\paragraph{Tightness}
The tightness of the lower bound is obvious as it suffices to take $k=1$ to get an equality between $\dmP$ and $\dPW$.

However, we do not know if the upper bound is tight.
The bound $\sqrt{6}$ can not be improved more than to $1+\sqrt{2}$, whose value is greater than $\sqrt{5.82}$.

Let us introduce the following example in $\R^d$. 
We fix $k=2d$ and $0<\epsilon<\sqrt{2}$.
The point cloud $P$ consists of $4d^2$ points located at the coordinates $(0,\cdots,0,\alpha,0,\cdots,0)$ with multiplicity $1$ when $\alpha=1$ or $\alpha=-1$ and multiplicity $2d-1$ when $\alpha=1+\sqrt{2}-\epsilon$ or $\alpha=\epsilon-1-\sqrt{2}$.

The following figure is its representation in dimension $2$ where the triangles have multiplicity $1$ and the circles have multiplicity $3$.
\begin{center}
\begin{tikzpicture}
\draw (-3,0) -- (3,0) (0,-3) -- (0,3);
\draw[blue,thick] (-2.414,0) circle (3pt);
\draw[red,thick] (-1.1,-.1) -- (-.9,-.1) -- (-1,.1) -- (-1.1,-.1);
\draw[blue,thick] (2.414,0) circle (3pt);
\draw[red,thick] (1.1,-.1) -- (.9,-.1) -- (1,.1) -- (1.1,-.1);
\draw[red,thick] (-.1,-1.1) -- (.1,-1.1) -- (0,-.9) -- (-.1,-1.1);
\draw[blue,thick] (0,-2.414) circle (3pt);
\draw[red,thick] (0,1.1) -- (-.1,.9) -- (.1,.9) -- (0,1.1);
\draw[blue,thick] (0,2.414) circle (3pt);
\draw[green,thick] (0,-.1) -- (0,.1) (-.1,0) -- (.1,0);
\end{tikzpicture}
\end{center}
The points are placed such that the $k$ nearest neighbors of any triangle are itself and the $k-1$ points located at the nearest circle.
These $k$ nearest neighbors are also the ones from the circles.

Let us now take a look to the value of the functions at the origin $o$.
Each of the $k$ nearest neighbors of $o$ are distance exactly $1$ from $o$.
This allows us to conclude that: $$\dmP(o)=1.$$

The construction induced that the structure is perfectly symmetric and the set of barycenters $W$ we consider in the witnessed $k$-distance contains exactly $2d$ points.
These points are located at the coordinates $(0,\cdots,0,\alpha,0,\cdots,0)$ where $\alpha=1+\frac{2d-1}{2d}(\sqrt{2}-\epsilon)$ or the opposite. 

Let $b$ be a member of $W$.
Thus we can compute its cell energy: 
\begin{align*}
E^C(b)&=\frac{1}{2d}\left[\left(\frac{2d-1}{2d}(\sqrt{2}-\epsilon)\right)^2+(2d-1)\left(\frac{1}{2d}(\sqrt{2}-\epsilon)\right)^2\right]\\
&=\frac{2d-1}{(2d)^3}\left[(2d-1)(\sqrt{2}-\epsilon)^2+(\sqrt{2}-\epsilon)^2\right]\\
&=\frac{2d-1}{(2d)^2}(\sqrt{2}-\epsilon)^2.
\end{align*}

All of the points of $W$ are located at the same distance to $o$.
Thus, the witnessed $k$-distance at the point $o$ is
\begin{align*}
  \dPW(o)^2
    &=E^C(b)+\left(1+\frac{d-1}{2d}(\sqrt{2}-\epsilon)\right)^2\\
    &=\frac{2d-1}{(2d)^2}(\sqrt{2} - \epsilon)^2+1+\frac{2d-1}{d}(\sqrt{2} - \epsilon)+\frac{(2d-1)^2}{(2d)^2}(\sqrt{2}-\epsilon)^2\\
    &=\frac{1}{2d}+\frac{2d-1}{2d}\left(1+2(\sqrt{2}-\epsilon)+(\sqrt{2}-\epsilon)^2\right)\\
    &=\frac{1}{2d}+\frac{2d-1}{2d}(1+\sqrt{2}-\epsilon)^2.
\end{align*}

Since we can take $\epsilon$ as small as we want and make the dimension grow, this relation assures us that we cannot find a better constant than $1+\sqrt{2}$ in Theorem~\ref{tWbound}.


\section{The Weighted Rips Filtration} 
\label{sWRips}

 Given a weighted set $(P,w)$ and the associated power distance $f$ (as in~\eqref{eq:power_distance}), one can introduce a generalization of the Rips filtration that is adapted to the weighted setting as has been done in~\cite{wkdGMM}.
  This construction allows us to approximate the persistence diagram of $\dm$ in some cases. 
  Moreover, we show that it is stable with respect to perturbation of the underlying sample (Theorem~\ref{tPQsamespace}) and that it gives a guaranteed approximation to the persistence diagram of the distance to an empirical measure (Theorem~\ref{thm:approximating_dPP_with_weighted_rips}). 
Furthermore, it has an interest of its own as it is stable for close weighted sets and can therefore be used as a shape signature.

Let us consider the sublevel set $f^{-1}(]-\infty,\alpha])$.
It is the union of the balls centered on the points $p$ of $P$ with radius $r_p(\alpha)=\sqrt{\alpha^2-w_p^2}$.
By convention, we consider that the ball is empty when the radius is imaginary.
We can define the nerve of this union:

\begin{definition}
Let $(P,w)$ be a weighted set in a metric space $\X$, then the \emph{weighted \v Cech complex} $C_\alpha(P,w)$ for parameter $\alpha$ is defined as the union of simplices $\sigma$ such that $\bigcap_{p\in\sigma}B(p,r_p(\alpha))\neq0$.
\end{definition}

  However, the \v Cech complex can be difficult to compute due the problem of testing if a collection of metric balls has a common intersection.
  Instead, we define a weighted version of the Rips complex, which only requires distance computations.

  \begin{definition}
  For a weighted set $(P,w)$ in a metric space $\X$, the \emph{weighted Rips complex} $R_\alpha(P,w)$ for a parameter $\alpha$ is the maximal simplicial complex whose 1-skeleton has an edge for each pair $(p,q)$ such that $\dX{p}{q} < r_p(\alpha) + r_q(\alpha)$.
  The \emph{weighted Rips filtration} is the sequence $\{R_\alpha(P,w)\}$ for all $\alpha\ge 0$.
  \end{definition}

Remark that if all weights are equal to $0$, we are in the classical case of balls with equal radii.
We use the weighted Rips filtration to approximate the weighted \v{C}ech filtration thanks to the following interleaving.
For simplicity, the notation $(P,w)$ indicating the point set $P$ with weights $w$ is omitted in the notation.

\begin{proposition}
If $(P,w)$ is a weighted set on a metric space $\X$, then for all $\alpha\in\R$:
$$C_\alpha\subseteq R_\alpha\subseteq C_{2\alpha}.$$
\end{proposition}

\begin{proof}
Let $\alpha$ be a real number.
The first inclusion is obtained by the definition of the weighted Rips complex that gives $C_\alpha\subseteq R_\alpha$.

For the other inclusion, let $\sigma$ be a simplex of $R_\alpha$.
We fix $p_0$ to be the point of $\sigma$ with the greatest weight. 
This implies especially that for any $p\in P$, $r_p(\alpha)\geq r_{p_0}(\alpha)$.

Since $\sigma\in R_\alpha$, we get that, for all $p$ and $q$ in $P$, we have $\dX{p}{q}\leq r_{p}(\alpha)+r_{q}(\alpha)$ with both radius real.
To prove that $\sigma\in C_{2\alpha}$ we need to prove that:
$$\bigcap_{p\in\sigma}B(p,r_{p}(2\alpha))\neq 0.$$

It will suffice to prove that $p_0$ belongs to this intersection.
For each $p\in\sigma$:
$$\dX{p}{p_0}\leq r_{p}(\alpha)+r_{p_0}(\alpha)\leq 2\ r_{p}(\alpha)= \sqrt{(2\alpha)^2-4w_{p}^2}\leq r_{p}(2\alpha).$$
\end{proof}

\paragraph{Stability\\}
The persistence diagram of a weighted Rips filtration $\{R_\alpha(P,w)\}$ is stable under small perturbations of the set $P$. 
It can thus be used in applications like signatures in the spirit of~\cite{ghssspCCGMO}.

Speaking of the persistence diagram of a weighted Rips filtration requires that the filtration is q-tame.
This is always the case when the set $P$ is compact as shown in the following proposition.

\begin{proposition}\label{pCompQTame}
Let $P$ be a subset of a metric space $\X$ and let $w:\X\to\R$ be a function.
If $P$ is compact, then $\{R_\alpha(P,w)\}_{\alpha\in\R}$ is q-tame.
\end{proposition}
This will be deduced from the following technical lemma.

\begin{lemma}\label{lHausInter}
Let $P$, $Q$ be two subsets of a metric space $\X$ and let $w:\X\to\R$ be a $t$-Lipschitz function.
Then $\Hom(\{R_\alpha(P,w)\})$ and $\Hom(\{R_\alpha(Q,w)\})$ are $\epsilon$-interleaved for $\epsilon=(1+t)d_H(P,Q)$.
\end{lemma}

\begin{proof}
We need to show that the there exists $\epsilon$-homomorphisms ${\pi_P}_*$ and ${\pi_Q}_*$ such that ${\pi_P}_*{\pi_Q}_*=1_{\Hom(R_\alpha(P,w))}^{2\epsilon}$ and ${\pi_Q}_*{\pi_P}_*=1_{\Hom(R_\alpha(Q,w))}^{2\epsilon}$.

To do so, we need three steps.
First, we build simplicial maps $R_\alpha(P,w)\to R_{\alpha+\epsilon}(Q,w)$ and $R_\alpha(Q,w)\to R_{\alpha+\epsilon}(P,w)$ for every $\alpha$.
Then, we show that these simplicial maps induce $\epsilon$-homomorphisms.
Finally, we show that the simplicial maps are contiguous and thus the two persistence modules are $\epsilon$-interleaved.

The simplicial maps $i_\alpha^\beta:R_\alpha(P,w)\to R_\beta(P,w)$ and $j_\alpha^\beta:R_\alpha(Q,w)\to R_\alpha(Q,w)$ for $\alpha<\beta$ are induced by the canonical inclusion. 
We consider two maps $\pi_P:Q\to P$ and $\pi_Q:P\to Q$ such that $\dX{p}{\pi_Q(p)}\leq d_H(P,Q)$ and $\dX{q}{\pi_P(q)}\leq d_H(P,Q)$ for any $p\in P$ and $q\in Q$.
By definition of the Hausdorff distances, such maps always exist.
Let us show that these maps induce simplicial maps.

Let us consider the function $\pi_P$ and let us fix $\alpha>0$.
Let $(q',q'')$ be an edge of $R_\alpha(Q,w)$. 
It means that $B(q',r_{q'}(\alpha))\cap B(q'',r_{q''}(\alpha))\neq\emptyset$.
Lemma~\ref{lPowerBalls} implies that $B(q,r_q(\alpha))\subset B(\pi_P(q),r_{\pi_P(q)}(\alpha+(1+t)d_H(P,Q)))$ for any $q\in Q$.
Thus, $(\pi_P(q'),\pi_(q''))$ is an edge of $R_{\alpha+\epsilon}(P,w)$ because:
$$B(\pi_P(q'),r_{q'}(\alpha+\epsilon))\cap B(\pi_P(q''),r_{q''}(\alpha+\epsilon))\supset B(q',r_{q'}(\alpha))\cap B(q'',r_{q''}(\alpha))\neq\emptyset.$$

As $R_{\alpha}(P,w)$ is a clique complex for any $\alpha$, this is sufficient to prove that $\pi_P$ induce a family of simplicial maps $\{{\pi_P}_\alpha^{\alpha+\epsilon}\}$.
The roles of $P$ and $Q$ are symmetric.
Therefore, the result holds for $\pi_Q$ as well.

Furthermore $\pi_P$ induces an $\epsilon$-homomorphism ${\pi_P}_*$ at the homology level. 
For any $\alpha<\beta$, $i_{\alpha+\epsilon}^{\beta+\epsilon}\circ{\pi_P}_{\alpha}^{\alpha+\epsilon}$=${\pi_P}_{\beta}^{\beta+\epsilon}\circ j_\alpha^\beta$ because the maps $i_{\alpha+\epsilon}^{\beta+\epsilon}$ and $j_\alpha^\epsilon$ are induced by the the canonical inclusion while the two others simplicial maps are induced by the same map $\pi_P:Q\to P$.
Hence the two compositions are contiguous and thus guarantees that ${\pi_P}_*$ is an $\epsilon$-homomorphism. 
Again, this results can be applied to $\pi_Q$ to get an $\epsilon$-homomorphism ${\pi_Q}_*$.

To prove that ${\pi_P}_*{\pi_Q}_*=1_{\Hom(R_\alpha(P,w))}^{2\epsilon}$, we prove that ${\pi_P}_\alpha^{\alpha+\epsilon}\circ{\pi_Q}_{\alpha-\epsilon}^\alpha$ and $i_{\alpha-\epsilon}^{\alpha+\epsilon}$ are contiguous for any $\alpha$.

Let us fix $\alpha$ and let $(p,p')$ be an edge of $R_{\alpha-\epsilon}(P,w)$.
By definition, $B(p,r_p(\alpha-\epsilon))\cap B(p',r_{p'}(\alpha-\epsilon))\neq\emptyset$.
Moreover, using Lemma~\ref{lPowerBalls} we get:
$$B(p,r_p(\alpha-\epsilon))\subset B(\pi_Q(p),r_{\pi_Q(p)}(\alpha))\subset B(\pi_P\circ\pi_Q(p),r_{\pi_P\circ\pi_Q(p)}(\alpha+\epsilon)).$$
The same holds for $p'$ and thus:
\begin{align*}
B(p,r_p(\alpha+\epsilon))\cap B(\pi_P\circ\pi_Q(p),r_{\pi_P\circ\pi_Q(p)}(\alpha+\epsilon))&\\
\cap B(p',r_{p'}(\alpha+\epsilon))\cap B(\pi_P\circ\pi_Q(p'),r_{\pi_P\circ\pi_Q(p')}(\alpha+\epsilon)) & \neq\emptyset.
\end{align*}
Thus the tetrahedron $\{i_{\alpha-\epsilon}^{\alpha+\epsilon}(p),i_{\alpha-\epsilon}^{\alpha+\epsilon}(p'),{\pi_P}_\alpha^{\alpha+\epsilon}\circ{\pi_Q}_{\alpha-\epsilon}^\alpha(p),{\pi_P}_\alpha^{\alpha+\epsilon}\circ{\pi_Q}_{\alpha-\epsilon}^\alpha(p')\}$ is in $C_{\alpha+\epsilon}(P,w)\subset R_{\alpha+\epsilon}(P,w)$.
Lemma~\ref{lem:contiguity_and_cliques} guarantees that ${\pi_P}_\alpha^{\alpha+\epsilon}\circ{\pi_Q}_{\alpha-\epsilon}^\alpha$ and $i_{\alpha-\epsilon}^{\alpha+\epsilon}$ are contiguous.

From before, $\{{\pi_P}_{\alpha}^{\alpha+\epsilon}\circ{\pi_Q}_{\alpha-\epsilon}^{\alpha}\}$ induces the $2\epsilon$-homomorphism ${\pi_P}_*{\pi_Q}_*$.
By definition, $\{i_{\alpha-\epsilon}^{\alpha+\epsilon}\}$ induces $1_{\Hom(R_\alpha(P,w))}^{2\epsilon}$.
By contiguity of the simplicial maps, we have equality of the $2\epsilon$-homomorphisms and therefore ${\pi_P}_*{\pi_Q}_*=1_{\Hom(R_\alpha(P,w))}^{2\epsilon}$.

By symmetry of the roles of $P$ and $Q$, $\{R_\alpha(P,w)\}$ and $\{R_\alpha(Q,w)\}$ are $\epsilon$-interleaved.
\end{proof}

\begin{proof}[Proof of Proposition~\ref{pCompQTame}]
We will show that, for any $\epsilon>0$, one can build a finite persistence module which is $\epsilon$-interleaved with the persistence module of $\{R_\alpha(P,w)\}$.
A finite persistence module is a fortiori locally finite and Theorem 4.19 of~\cite{sspmCDGO} induces the q-tameness of $\{R_\alpha(P,w)\}$.

Let us fix $\epsilon>0$.
$P$ is compact.
As a consequence, there exists a finite point set $Q$ of $P$ such that $d_H(P,Q)\leq\frac{\epsilon}{1+t}$.
The persistence module of $\{R_\alpha(Q,w)\}$ is finite and therefore locally finite.
Moreover, using Lemma~\ref{lHausInter}, $\{R_\alpha(Q,w)\}$ and $\{R_\alpha(P,w)\}$ are $\epsilon$-interleaved.
Hence $\{R_\alpha(P,\alpha)\}$ is q-tame using Theorem 4.19 of~\cite{sspmCDGO} induces the q-tameness of $\{R_\alpha(P,w)\}$.
\end{proof}

Notice that the simplicial maps $\pi_P$ and $\pi_Q$ are not necessarily uniquely defined.
However, if $\pi_P$ and $\pi_P'$ are two maps verifying the construction property, then the induced simplicial maps are contiguous and therefore the induced homomorphisms are identical.

The persistence diagrams of weighted Rips filtrations are related by the following:

\begin{theorem}\label{tPQsamespace}
Let $P$ and $Q$ be two compact subsets of a metric space $\X$.
Let $w:\X\to\R$ be a $t$-Lipschitz function.
Then,
\[
  d_B(\Dgm{\{R_\alpha(P,w)\}},\Dgm{\{R_\alpha(Q,w)\}})\leq(1+t)d_H(P,Q).
\]
\end{theorem}

\begin{proof}
$P$ and $Q$ are two compact sets and thus the diagrams are well-defined thanks to Proposition~\ref{pCompQTame} that guarantees the q-tameness of the filtrations. 
Lemma~\ref{lHausInter} implies that $\Hom(\{R_\alpha(P,w)\})$ and $\Hom(\{R_\alpha(Q,w)\})$ are $(1+t)d_H(P,Q)$-interleaved.
The relation between the persistence diagrams is then obtained by applying Theorem~\ref{tModStability}.
\end{proof}

\begin{remark}
When $P$ and $Q$ are two compact metric spaces, Theorem~\ref{tPQsamespace} can be extended using the notion of correspondence as in~\cite{psgcCDO}.
Notice that the correspondence has to induce bounded distortion on the weights as well as on the distances.
\end{remark}

\paragraph{Approximation\\}
To use the weighted Rips filtration to approximate the persistence diagram of the distance to a measure, we need to restrict the class of spaces considered.
If the intersection of any finite number of balls in $\X$ is either contractible or empty, $\X$ is said to have the \emph{good cover property}.
Then the \v{C}ech complex has the same homology as the union of balls, of which it is the nerve, by the Nerve Theorem~\cite{atH}.
We can also compute the persistence diagram thanks to the Persistent Nerve Lemma~\cite{tpbresCO}.
We obtain an approximation of $\Dgm{\dmP}$ using the weighted Rips filtration.

\begin{theorem}\label{thm:approximating_dPP_with_weighted_rips}
Let $\X$ be a triangulable metric space with the good cover property and let $P$ be a finite point set of $\X$, then on a logarithmic scale:
$$\dbl(\Dgm{\dmP},\Dgm{\{R_\alpha(P,\dmP)\}})\leq \ln(2\sqrt{5}).$$
\end{theorem}

\begin{proof}
Given that $\X$ is triangulable, we know that the sublevel sets filtration of $\dmP$ is $q$-tame by Proposition~\ref{pQtame}.
The persistence diagram $\Dgm{\dmP}$ is thus well-defined.
Recall that $\dmP$ is a $1$-Lipschitz function (see Proposition~\ref{pLipschitz}).
$P$ is a compact subset of $\X$ and therefore $\Dgm{R_\alpha(P,\dmP)}$ is well-defined according to Proposition~\ref{pCompQTame}.

We approximate $\dmP$ with $\dPP$. 
The result of Theorem~\ref{tPbound} gives us a $\sqrt{5}$ multiplicative interleaving.
For any $\alpha\in\R$,
\[
  \dmP(]-\infty,\alpha])\subset\dPP(]-\infty,\sqrt{2}\alpha])\subset\dmP(]-\infty,\sqrt{10}\ \dPP]).
\]
So, Theorem~\ref{tstability} implies
\[
  \dbl(\Dgm{\dmP},\Dgm{\dPP})\leq\ln(\sqrt{5}).
\]

By the Persistent Nerve Lemma, the sublevel sets filtration of $\dPP$ (a union of balls of increasing radii) has the same persistent homology as nerve filtration.
Thus, we can use weighted Rips filtration to approximate the persistence diagram:
$$\dbl(\Dgm{\dPP},\Dgm{\{R_\alpha(P,\dmP)\}}\leq \ln(2).$$

The triangle inequality for the bottleneck distance gives the desired inequality.
\end{proof}

\section{The Sparse Weighted Rips filtration} 
\label{sec:sparse_rips}

The weighted Rips filtration presented in the previous section has the desired approximation guarantees, but like the Rips filtration for unweighted points, it usually grows too large to be computed in full.
In~\cite{lsavrfS}, it was shown how to construct a filtration $\{S_\alpha\}$ called the \emph{sparse Rips filtration} that gives a provably good approximation to the Rips filtration and has size linear in the number of points for metrics with constant doubling dimension (see Section~\ref{sec:sparse_rips_revisited} for the construction).
Specifically, for a user-defined parameter $\e$, the log-bottleneck distance between the persistence diagrams of the Sparse Rips filtration and the Rips filtration is at most $\e$.
The goal of this section is to extend that result to weighted Rips filtrations.

The sparse Rips filtration cannot be used directly here, since the power distance does not induce a metric.
Indeed, even the case of setting all weights to some large constant yields a persistence diagram that is far from the persistence diagram of the Rips filtration of any metric space.
This follows because individual points in a Rips filtration appear at time zero, but this is not the case in the weighted Rips filtration.

Even if one were to construct a metric whose Rips filtration exactly matched that of the weighted Rips filtration, there are simple examples where that metric would necessarily have very high doubling dimension, making previous methods unsuitable.
For example, consider a set of points in the unit interval $[0,1]$, with a constant weight function that assigns a weight of $1$ to every point.
Although the points lie in a $1$-dimensional space, the weighted distance function has doubling dimension $\log n$ because all of the points are within a weighted distance of $2$, whereas every pair has weighted distance at least $1$.
So the doubling constant would be $n$ and the doubling dimension would be $\log n$ despite that the input was $1$-dimensional.
Thus, any construction that depends on low doubling dimension will blowup when confronted with such weighted examples.



For the rest of this section, we fix a weighted point set $P$ in a metric space $\X$, where the weight function $w:\X\to\R$ is $t$-Lipschitz, for some constant $t$.
To simplify notation, we let $R_\alpha$ denote the weighted Rips complex $R_\alpha(P,w)$.

The \emph{sparse weighted Rips filtration}, $\{T_\alpha\}$, is defined as
\[
  T_\alpha = S_\alpha \cap R_\alpha.
\]
The (unweighted) sparse Rips filtration $\{S_\alpha\}$ captures the underlying metric space and the weighted Rips filtration $\{R_\alpha\}$ captures the structure of the sublevel sets of the power distance function.
Computing $\{T_\alpha\}$ can be done efficiently by first computing $\{S_\alpha\}$ and then reordering the simplices according to the birth time in $\{R_\alpha\}$.
This is equivalent to filtering the complex $S_\infty$.
Note that the sparsification depends only on the metric, and not on the weights.
Thus, the same sparse Rips complex can be used as the underlying complex for different weight functions.
We also simplify the construction of $\{S_\alpha\}$ by using a furthest point sampling instead of the more complex structure of net tree.



The technical challenge is to relate the persistence diagram of this new filtration to the persistence diagram of the weighted Rips filtration as in the following theorem.

\begin{theorem}\label{thm:sparse_weighted_rips}
  Let $(P,w)$, be a finite, weighted subset of a metric space $\X$ with $t$-Lipschitz weights.
  Let $\e<1$ be a fixed constant used in the construction of the sparse weighted Rips filtration $\{T_\alpha\}$.
  Then,
  \[
    \dbl(\Dgm{\{T_\alpha\}}, \Dgm{\{R_\alpha\}}) \le \ln \left(\frac{1 + \sqrt{1+t^2}\;\e}{1-\e}\right).
  \]
\end{theorem}

Since these filtrations are not interleaved, the only hope is to find an interleaving of the persistence modules, which requires finding suitable homomorphisms between the homology groups of the different filtrations.
After detailing the construction of the sparse Rips filtration with the furthest point sampling, the rest of this section proves Theorem~\ref{thm:sparse_weighted_rips}.
\subsection{Sparse Rips complexes}\label{sec:sparse_rips_revisited}

Let $(p_1,\ldots, p_n)$ be a greedy permutation of the points $P$ in a finite metric space $\X$.
That is, $p_i = \argmax_{p\in P\setminus P_{i-1}} \dX{p}{P_{i-1}}$, where $P_{i-1} = \{p_1,\ldots,p_{i-1}\}$ is the $(i-1)^\mathrm{st}$ prefix.
  We define the \emph{insertion radius} $\ir_{p_i}$ of point $p_i$ to be
  \[
    \ir_{p_i} = \dX{p_i}{P_{i-1}}.
  \]
  
To avoid excessive superscripts, we write $\ir_i$ in place of $\ir_{p_i}$ when we know the index of $p_i$.
We adopt the convention that $\ir_1 = \infty$ and $\ir_{n+1} = 0$.
The greedy permutation has the nice property that each prefix $P_i$ is a \emph{$\ir_i$-net} in the sense that 
  \begin{enumerate}
    \item     $\dX{p}{P_i}\leq\ir_i$ for all $p\in P$.
    \item     $\dX{p}{q} \ge \ir_i$  for all $p,q\in P_i$.
  \end{enumerate}
  We extend these nets to an arbitrary parameter $\gamma$ as
  \begin{align*}
    \net_{\gamma} &= \{p\in P \mid \ir_p > \gamma\}.\\
    \cl\net_{\gamma} &= \{p\in P \mid \ir_p \ge \gamma\}.   
  \end{align*}
  Note that for all $p\in P$, $\dX{p}{\net_{\gamma}}\le \gamma$ and $\dX{p}{\cl\net_{\gamma}}< \gamma$.
  
One way to get a sparse Rips-like filtration is to take a union of Rips complexes on the nets $N_\gamma$.
However, this can add significant noise to the persistence diagram compared to the Rips filtrations.
This noise can be diminished by a careful perturbation of the distance.
For a point $p$, the perturbation varies with the scale and is defined as follows:
  \[
    s_p(\alpha) = 
      \begin{cases}
        0 & \text{if } \alpha \le \frac{\ir_p}{\e}\\
        \alpha - \frac{\ir_p}{\e} & \text{if } \frac{\ir_p}{\e} < \alpha < \frac{\ir_p}{\e(1-\e)}\\
        \e \alpha & \text{if }  \frac{\ir_p}{\e(1-\e)} \le \alpha 
      \end{cases}
  \]
  
\begin{center}
\begin{tikzpicture}[scale=.5]
\draw (0,0) -- (10,0);
\draw (0,0) -- (0,6);
\draw[red,thick] (0,0) -- (3,0) -- (6,3) -- (10,5);
\draw[red,dashed] (0,0) -- (6,3);
\node[below] at (0,0) {$0$};
\node[left] at (0,0) {$0$};
\node[below] at (3,0) {$\frac{\lambda_p}{\epsilon}$};
\node[below] at (6,0) {$\frac{\lambda_p}{\epsilon(1-\epsilon)}$};
\node[above] at (9.5,5) {$s_p=\epsilon\alpha$};
\end{tikzpicture}
\end{center}
  
  Note that $s_p$ is $1$-Lipschitz.
  The resulting perturbed distance is defined as
  \[
    \ff_{\alpha}(p,q) = \dX{p}{q} + s_p(\alpha) + s_q(\alpha).
  \]
  
  For any fixed $p$ and $q$, the Lipschitz property of $s_p$ and $s_q$ implies that for all $\alpha\le \beta$:
  \[
    \ff_\alpha(p,q) \le \ff_\beta(p,q) + 2(\beta-\alpha).
  \]

\begin{definition}
Given the nets $N_\gamma$ and the distance function $f_\alpha$, we define the \emph{sparse Rips complex} at scale $\alpha$ as
  \[
    Q_\alpha = \{\sigma\subset\cl\net_{\e(1-\e)\alpha}\mid\forall p,q\in\sigma,\ \ff_\alpha(p,q)\leq 2\alpha\}.
  \]
  \end{definition}
  
  On its own, the sequence of complexes $\{Q_\alpha\}$ does not form a filtration.
  However, we can build a natural filtration by defining 
  
\begin{definition}
The \emph{sparse Rips filtration} is defined as:
  \[
    S_\beta = \bigcup_{\alpha\le \beta} Q_\alpha.
  \]
\end{definition}

\subsection{Projection onto Nets} 
\label{sub:projection_onto_nets}
To relate sparse Rips complexes with Rips complexes, we build a collection of projections of the points onto the nets.
  \[
    \proj_\alpha(p) = 
      \begin{cases}
        p & \text{if } p\in \net_{\e(1-\e)\alpha}\\
        \argmin_{q\in \net_{\e\alpha}} \dX{p}{q} & \text{otherwise}
      \end{cases}
  \]
  For any scale $\alpha$, the projection $\proj_\alpha$ maps the points of $P$ to the net $\net_{\e(1-\e)\alpha}$.
  Note that $\proj_\alpha$ is a retraction onto $\net_{\e(1-\e)\alpha}$.
  
  The following are the four main lemmas we will use with respect to the perturbed distance functions and projections.
  The projections will be used extensively to induce maps between simplicial complexes.
  
First, we prove that edges do not disappear as the filtration grows.  
  \begin{lemma}\label{lem:filtration}
    If $\ff_\alpha(p,q) \le 2 \alpha \le 2 \beta$ then $\ff_\beta(p,q)\le 2 \beta$.
  \end{lemma}
  \begin{proof}
    The proof follows from the definitions $\ff_\alpha$ and $\ff_\beta$, the Lipschitz property of the perturbations $s_p$ and $s_q$, and the hypothesis as follows.
    \begin{align*}
      \ff_\beta(p,q) 
        &=   \dX{p}{q} + s_p(\beta) + s_q(\beta)\\ 
        &\le \dX{p}{q} + s_p(\alpha) + (\beta-\alpha) + s_q(\alpha) + (\beta-\alpha)\\ 
        &=   \ff_\alpha(p,q) + 2(\beta-\alpha)\\ 
        &\le 2 \alpha + 2(\beta-\alpha)\\ 
        &=   2 \beta.\qedhere
    \end{align*}
  \end{proof}
 
Next, we show that the distance between a point and its projection is at most the change in the perturbed distance.
  \begin{lemma}\label{lem:distance_to_proj} 
    For all $q\in P$, $\dX{q}{\proj_\alpha(q)} \le s_q(\alpha) - s_{\proj_\alpha(q)}(\alpha)$, and in particular, $\dX{q}{\proj_\alpha(q)}\le \e \alpha$.
  \end{lemma}
  \begin{proof}
    Both statements are trivial if $q\in \net_{\e(1-\e)\alpha}$, because that would imply that $\proj_\alpha(q) = q$.
    So, we may assume that $\proj_\alpha(q)$ is the nearest point to $q$ in $\net_{\e \alpha}$
    It follows that 
    \[
      \dX{q}{\proj_\alpha(q)} \le \e \alpha.
    \]
    Moreover, $\ir_q \le \e(1-\e)\alpha$, and thus $s_q(\alpha) = \e \alpha$.    
    Also, since $\proj_\alpha(q)\in \net_{\e \alpha}$, it must be that $\ir_{\proj_\alpha(q)} > \e \alpha$ and so $s_{\proj_\alpha(q)} = 0$.
    Combining these statements, we get
    \[
      \dX{q}{\proj_\alpha} \le \e \alpha = s_q(\alpha) - s_{\proj_\alpha(q)}(\alpha).
    \]
  \end{proof}

Now, we prove that replacing a point with its projection does not increase the perturbed distance.
  \begin{lemma}\label{lem:distances_dont_grow}
    For all $p,q\in P$ and all $\alpha\ge 0$, $\ff_\alpha(p,\proj_\alpha(q)) \le \ff_\alpha(p,q)$.\qedhere
  \end{lemma}
  \begin{proof}
    The statement follows from the definition of $\ff_\alpha$, the triangle inequality, and Lemma~\ref{lem:distance_to_proj} as follows.
    \begin{align*}
      \ff_\alpha(p, \proj_\alpha(q)) 
        &= \dX{p}{\proj_\alpha(q)} + s_p(\alpha) + s_{\proj_\alpha(q)}(\alpha)\\ 
        &\le \dX{p}{q} + \dX{q}{\proj_\alpha(q)} + s_p(\alpha) + s_{\proj_\alpha(q)}(\alpha)\\ 
        &\le \dX{p}{q} + s_p(\alpha) + s_q(\alpha)\\ 
        &= \ff_\alpha(p,q).\qedhere
    \end{align*}
  \end{proof}

We want to use the sparse Rips filtration in the weighted setting. 
Recall that for a weighted point $p$, $r_p(\alpha)=\sqrt{\alpha^2-w_p^2}$.

We consider the effect on the ``edge lengths'' when projecting the endpoints of an edge to nearby points.
This is the situation when we project the metric onto an $\e$-net. 
The following lemma guarantees that a ball centered at the image of the projection quickly covers the ball centerd at the original point.
It is a similar approach to the Proposition~\ref{pPowerBalls}.
  
\begin{center}
\begin{tikzpicture}
\draw[fill] (0,0) circle (1pt);
\draw[fill] (.5,0) circle (1pt);
\node[below] at (0,0) {$p$};
\draw (.5,-.5) node {$q$};

\draw (0,0) circle (55pt);
\draw[dashed] (0,0) circle (70pt);
\draw[red] (.5,0) circle (85pt);
\draw[->] (0,0) -- node[left] {$r_p(\alpha)$} (.2,1.9) ;
\draw[->,dashed] (.2,1.9) -- node[right] {$\epsilon\alpha$} (.25,2.4);
\draw[->,dashed] (0,0) -- node[above] {$\epsilon\alpha$} (.5,0);
\draw[->,red] (.5,0) -- node[above] {$r_q((1+\epsilon+t\epsilon)\alpha)$} (3.5,.3);
\end{tikzpicture}
\end{center}



\subsection{Sometimes the projections induce contiguous simplicial maps} 
\label{sub:sometimes_the_projections_induce_contiguous_simplicial_maps}

  In this section, we look at the maps between simplicial complexes that are induced by the projection functions $\proj_\alpha$.
  We are most interested in the case when a pair of projections $\proj_\alpha$ and $\proj_\beta$ induce contiguous simplicial maps between sparse Rips complexes (Lemma~\ref{lem:proj_and_Q}) or weighted Rips complexes (Lemma~\ref{lem:proj_and_R}).
  First, we need a couple lemmas that describe the effect of different projections on the endpoints of an edge in sparse or weighted Rips complexes.

  \begin{lemma}\label{lem:proj_and_Q_for_one_edge}
    Let $\alpha$, $\beta$, $\gamma$, and $i$ be such that $\frac{\ir_{i+1}}{\e(1-\e)}\le \alpha \le \beta \le \gamma \le \frac{\ir_i}{\e(1-\e)}$.
    If an edge $(p,q)$ is in $Q_\rho$ for some $\rho\le \gamma$ then the edge $(\proj_\alpha(p), \proj_\beta(q)) \in Q_\gamma$.
  \end{lemma}
  \begin{proof}
    First, it is easy to check that the conditions on $\alpha$, $\beta$, $\gamma$, and $i$ imply that $\proj_\alpha(p)$ and $\proj_\beta(q) $ are in $\cl\net_{\e(1-\e)\gamma}$, which is the vertex set of $Q_\gamma$.
    So, it will suffice to prove that $\ff_\gamma(\proj_\alpha(p), \proj_\beta(q))\le 2\gamma$.
    Next we consider three cases depending on the value of $\rho$ in relation to $\alpha$ and $\beta$.
  
    \noindent\textbf{Case 1:} If $\alpha,\beta\le \rho$ then $\proj_\alpha(p) = p$ and $\proj_\beta(q) = q$.
    So, using Lemma~\ref{lem:filtration} and the assumption $\rho\le \gamma$, we see that $\ff_\gamma(\proj_\alpha(p), \proj_\beta(q)) = \ff_\gamma(p,q) \le 2\gamma$.
  
    \noindent\textbf{Case 2:} If $\alpha\le \rho < \beta$ then $\proj_\alpha(p) = p$ and Lemma~\ref{lem:filtration} implies that $\ff_\beta(p,q)\le 2\beta$.
    \begin{align*}
      \ff_\gamma(\proj_\alpha(p), \proj_\beta(q))
        &= \ff_\gamma(p, \proj_\beta(q)) \\ 
        &\le \ff_\beta(p, \proj_\beta(q)) + 2(\gamma-\beta)\\ 
        &\le \ff_\beta(p, q) + 2(\gamma-\beta)\\ 
        &\le 2\gamma. 
    \end{align*}
  
    \noindent\textbf{Case 3:} If $\rho < \alpha, \beta$ then Lemma~\ref{lem:filtration} implies that $\ff_\alpha(p,q)\le 2\alpha$.
    \begin{align*}
      \ff_\gamma(\proj_\alpha(p), \proj_\beta(q))
        &\le \ff_\beta(\proj_\alpha(p), \proj_\beta(q)) + 2(\gamma-\beta)\\ 
        &\le \ff_\beta(\proj_\alpha(p), q) + 2(\gamma-\beta)\\ 
        &\le \ff_\alpha(\proj_\alpha(p), q) + 2(\gamma-\beta) + 2(\beta-\alpha)\\ 
        &\le \ff_\alpha(p, q) + 2(\gamma-\beta) + 2(\beta-\alpha)\\ 
        &\le 2\gamma. \qedhere 
    \end{align*}
  \end{proof}

  \begin{lemma}\label{lem:proj_and_R_for_one_edge}
Let $(p,q)$ be an edge of $R_\delta$ with $\alpha, \beta \le \frac{\delta}{1+\e}$, then $(\proj_\alpha(p), \proj_\beta(q))\in R_{\kappa\delta}$, where $\kappa= \frac{1 + \sqrt{1+t^2}\;\e}{1-\e}$.
  \end{lemma}
  \begin{proof}
    First, note that the projection functions satisfy the following inequalities.
    \begin{align*}
      \dX{p}{\proj_\alpha(p)} &\le \e \alpha \le \frac{\e \delta}{1-\e} \\
      \dX{q}{\proj_\beta(q)} &\le \e \beta \le \frac{\e \delta}{1-\e}
    \end{align*}
    So, by applying the triangle inequality, the definition of an edge in $R_\delta$, and Lemma~\ref{lPowerBalls}, we get the following.
    \begin{align*}
      \dX{\proj_\alpha(p)}{\proj_\beta(q)} 
        &\le \dX{p}{q} + \frac{2\e\delta}{1-\e}\\
        &\le \left(r_p(\delta) + \frac{\e\delta}{1-\e}\right) + \left(r_q(\delta) + \frac{\e\delta}{1-\e}\right)\\
        &\le \left(r_p\left(\frac{\delta}{1-\e}\right) + \frac{\e\delta}{1-\e}\right) + \left(r_q\left(\frac{\delta}{1-\e}\right) + \frac{\e\delta}{1-\e}\right)\\
        &\le r_{\proj_\alpha(p)}(\kappa\delta) + r_{\proj_\beta(q)}(\kappa\delta).
    \end{align*}
    This is precisely the condition necessary to guarantee that $(\proj_\alpha(p), \proj_\alpha(q))\in R_{\kappa \delta}$ as desired.
  \end{proof}

  The following two lemmas follow easily from repeated application of the preceding lemmas.

  \begin{lemma}\label{lem:proj_and_Q}
    Two projections $\proj_\alpha$ and $\proj_\beta$ induce contiguous simplicial maps from $Q_\rho\to Q_\beta$ whenever $\rho \le \beta$ and there exists $i$ so that $\frac{\ir_{i+1}}{\e(1-\e)}\le \alpha \le \beta \le \frac{\ir_{i}}{\e(1-\e)}$.
  \end{lemma}
  
\begin{proof}
Let us fix $\rho\leq\beta$ and take $(p,q)$ an edge from $Q_\rho$. 
Given that $Q_\rho$ and $Q_\beta$ are cliques complexes, we can get the result from Lemma~\ref{lem:contiguity_and_cliques} if we show that the tetrahedron $\{\proj_\alpha(p),\proj_\alpha(q),\proj_\beta(p),\proj_\beta(q)\}$ is in $Q_\beta$.
We only need to prove that all edges of the tetrahedron belongs to $Q_\beta$.

We apply Lemma~\ref{lem:proj_and_Q_for_one_edge}, while replacing $\gamma$ by $\beta$ and $\beta$ by $\alpha$. Thus we obtain $(\proj_\alpha(p),\proj_\alpha(q))\in Q_\beta$.
Let us repeat this operation with $\alpha=\beta=\gamma$ thus we get $(\proj_\beta(p),\proj_\beta(q))\in Q_\beta$.
The last two edges are given by replacing $\gamma$ by $\beta$ and choosing correctly the role of $p$ and $q$.
\end{proof}

  \begin{lemma}\label{lem:proj_and_R}
    Two projections $\proj_\alpha$ and $\proj_\beta$ induce contiguous simplicial maps from $R_\delta\to R_{\kappa\delta}$, where $\kappa = \frac{1 + \sqrt{1+t^2}\;\e}{1-\e}$ whenever $\alpha,\beta \le \frac{\delta}{1-\e}$.    
  \end{lemma}
  
\begin{proof}
The previous proof can be applied to get the result, while replacing Lemma~\ref{lem:proj_and_Q_for_one_edge} by Lemma~\ref{lem:proj_and_R_for_one_edge}.
\end{proof}



\subsection{Sparse filtrations and power distance functions} 
\label{sec:rips_interleaving}

We define a sparse filtration that gives a good approximation to the weighted Rips filtration $\{R_\alpha\}$ in terms of persistent homology.
  It is simply the intersection of the weighted Rips complex and the union of sparse Rips complexes at different scales.
  \[
    T_\alpha = R_\alpha \cap S_\alpha.
  \]

  Our main goal is to show that the filtration $\{T_\alpha\}$ has a persistence diagram that is similar to that of $\{R_\alpha\}$.
  To do this we will demonstrate a multiplicative interleaving between these filtrations, where the interleaving constant is 
  \[
    \kappa =\frac{1 + \sqrt{1+t^2}\;\e}{1-\e}.
  \]
  Specifically, we show that for all $\alpha\ge 0$, the following diagram commutes at the homology level. 
  
  \[
    \xymatrix{
      R_{\alpha} \ar[dr]^{\proj_{\frac{\alpha}{1-\e}}} \ar@{^{(}->}[r] & R_{\kappa\alpha} \\
      T_{\alpha} \ar@{^{(}->}[r] \ar@{^{(}->}[u] & T_{\kappa\alpha} \ar@{^{(}->}[u]
    }
  \]

  We first need to check that the projection $\proj_{\frac{\alpha}{1-\e}}$ indeed induces a simplicial map from $R_\delta$ to $T_{\kappa\delta}$.

  \begin{lemma}\label{lem:proj_is_simplicial_R_to_T}
    For all $\alpha> 0$, the projection $\proj_{\frac{\alpha}{1-\e}}$ induces a simplicial map from $R_\alpha\to T_{\kappa\alpha}$, where $\kappa = \frac{1 + \sqrt{1+t^2}\;\e}{1-\e}$.
  \end{lemma}
  \begin{proof}
    We show that for each edge $(p,q)\in R_\alpha$, there is a corresponding edge $(\proj_{\frac{\alpha}{1-\e}}(p), \proj_{\frac{\alpha}{1-\e}}(q)) \in R_{\kappa\alpha}\cap Q_{\frac{\alpha}{1-\e}}$.
    Since the latter complex is a clique complex, this will imply that for all $\sigma\in R_\alpha$, we have $\proj_{\frac{\alpha}{1-\e}}(\sigma)\in R_{\kappa\alpha}\cap Q_{\frac{\alpha}{1-\e}} \subseteq T_{\kappa\alpha}$ as desired.
    First, $(\proj_{\frac{\alpha}{1-\e}}(p), \proj_{\frac{\alpha}{1-\e}}(q)) \in R_{\kappa\alpha}$ as a direct consequence of Lemma~\ref{lem:proj_and_R}.

    Next, we need to show that $(\proj_{\frac{\alpha}{1-\e}}(p), \proj_{\frac{\alpha}{1-\e}}(q)) \in Q_{\frac{\alpha}{1-\e}}$.
    It suffices to show that $\ff_{\frac{\alpha}{1-\e}}(\proj_{\frac{\alpha}{1-\e}}(p), \proj_{\frac{\alpha}{1-\e}}(q))\le \frac{2 \alpha}{1-\e}$.
    \begin{align*}
      \ff_{\frac{\alpha}{1-\e}}(\proj_{\frac{\alpha}{1-\e}}(p), \proj_{\frac{\alpha}{1-\e}}(q))
        &\le \ff_{\frac{\alpha}{1-\e}}(p, q) \\ 
        &= \dX{p}{q} + s_p({\frac{\alpha}{1-\e}}) + s_q({\frac{\alpha}{1-\e}}) \\
        &\le \dX{p}{q} + \frac{2\e\alpha}{1-\e}\\
        &\le 2 \alpha + \frac{2\e\alpha}{1-\e}\\
        &= \frac{2 \alpha}{1-\e}\qedhere
    \end{align*}
  \end{proof}
  
  Now, we give conditions for when two projections induce contiguous simplicial maps between the sparse weighted Rips complexes $T_\delta$ and $T_{\kappa\delta}$.

  \begin{lemma}\label{lem:proj_and_T}
    Two projections $\proj_\alpha$ and $\proj_\beta$ induce contiguous simplicial maps from $T_\delta\to T_{\kappa\delta}$, where $\kappa = \frac{1 + \sqrt{1+t^2}\;\e}{1-\e}$ whenever $\alpha,\beta \le \frac{\delta}{1-\e}$ and there exists $i$ so that $\frac{\ir_{i+1}}{\e(1-\e)}\le \alpha \le \beta \le \frac{\ir_{i}}{\e(1-\e)}$.
  \end{lemma}
  \begin{proof}
    We simply observe that for any $\sigma\in T_\delta$, $\sigma\in Q_\rho$ for some $\rho\le \delta$.
    If $\rho\le \beta$ then Lemma~\ref{lem:proj_and_Q} implies $\proj_\alpha(\sigma) \cup \proj_\beta(\sigma) \in Q_\beta$.
    Otherwise $\proj_\alpha(\sigma) \cup \proj_\beta(\sigma) = \sigma \in Q_\rho$.
    So in either case, we have $\proj_\alpha(\sigma) \cup \proj_\beta(\sigma) \in S_{\delta \gamma}$.
    Now, by Lemma~\ref{lem:proj_and_R}, we have that $\proj_\alpha(\sigma) \cup \proj_\beta(\sigma) \in R_{\kappa\delta}$.
    So, we have that $\proj_\alpha(\sigma) \cup \proj_\beta(\sigma) \in R_{\kappa\delta} \cap S_{\kappa\delta} = T_{\kappa\delta}$ as desired.
  \end{proof}

  We can now give the proof of the interleaving which will imply the desired approximation of the persistent homology.

  \begin{lemma}\label{lem:interleaving}
    For all $\alpha>0$, the following diagram commutes the homology level.
    \[
      \xymatrix{
        R_{\alpha} \ar[dr]^{\proj_{\frac{\alpha}{1-\e}}} \ar@{^{(}->}[r] & R_{\kappa\alpha} \\
        T_{\alpha} \ar@{^{(}->}[r] \ar@{^{(}->}[u] & T_{\kappa\alpha} \ar@{^{(}->}[u]
      }
    \]
  \end{lemma}
  \begin{proof}
    By Lemma~\ref{lem:proj_and_R}, the projection $\proj_{\frac{\alpha}{1-\e}}$ and the inclusion $\proj_0$ are contiguous and thus produce identical homomorphisms at the homology level.
    For the lower triangle it will suffice to show that homomorphism induced by $\proj_{\frac{\alpha}{1-\e}}$ commutes with that produced by the inclusion $\proj_0$.
    Let $\phi_i = \proj_{\frac{\ir_i}{1-\e}}$ for $i=1,\ldots,n+1$.
    Now, Lemma~\ref{lem:proj_and_T} implies that $\phi_i$ and $\phi_{i+1}$ are contiguous.
    So, choosing $k$ such that $\ir_k \le \e\alpha <\ir_{k-1}$, we can apply Lemma~\ref{lem:proj_and_T} repeatedly to conclude that
    \[
      \proj_{0\star} = \phi_{n+1\star} =\phi_{n\star} = \cdots = \phi_{k\star} = \proj_{\frac{\alpha}{1-\e}\star}.\qedhere
    \]
  \end{proof}


\section{Numerical illustration} 
\label{sNumeric}

In this section, we illustrate our results three different perspectives:
the quality of the approximation,
the stability of the diagrams with respect to noise, and
the size of the filtration after sparsification.

We used the ANN library~\cite{annMS} for the $k$-nearest neighbors search and code from Zomorodian following~\cite{cphCZ} for the persistence.
The topology of the union of balls is acquired through the $\alpha$-shapes implementation from the CGAL library~\cite{cgalAlphaShapes3D}.

\paragraph{Datasets\\}
For the first two parts, we consider the set of points in $\R^3$ obtained by sampling regularly the skeleton of the unit cube with 116 points.
Then we add four noise points in the center of four of its faces such that two opposite faces are empty.

\begin{figure}[!ht]
\centering
\includegraphics[height=10em]{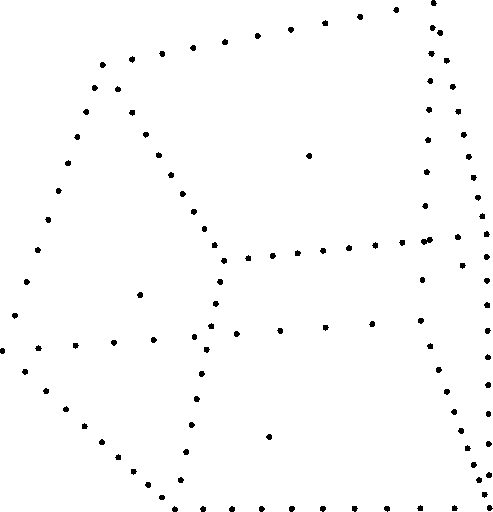}
\caption{Skeleton of a cube with outliers}
\end{figure}

We would like to compute the persistence diagram of the skeleton of the cube. 
We write this diagram $\Dgm{Skel}$.
It contains five homology classes in dimension 1 and one in dimension 2, and it has the barcode representation given in Figure~\ref{fDgmSkel}.

\begin{figure}[!ht]
\centering
\includegraphics[height=10em]{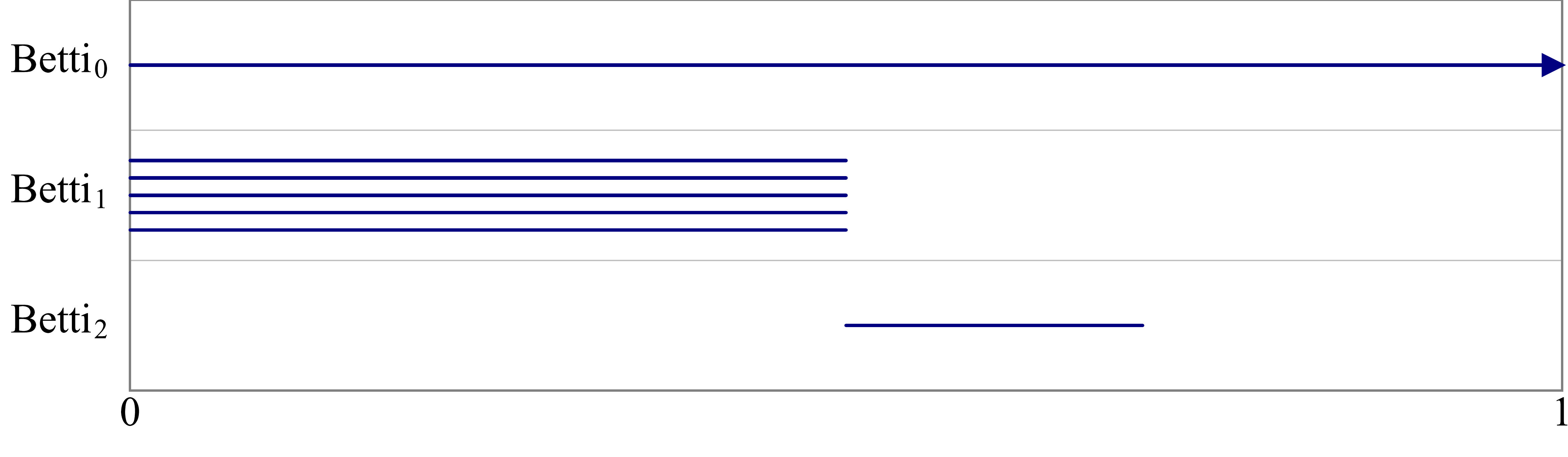}
\caption{Persistence diagram of a cube skeleton without noise}\label{fDgmSkel}
\end{figure}

For sparsification, we use a slightly bigger dataset composed of 10000 points regularly distributed on a curve rolled around a torus.
The point set is shown on Figure~\ref{fSpiral}.

\begin{figure}[!ht]
\centering
\includegraphics[height=10em]{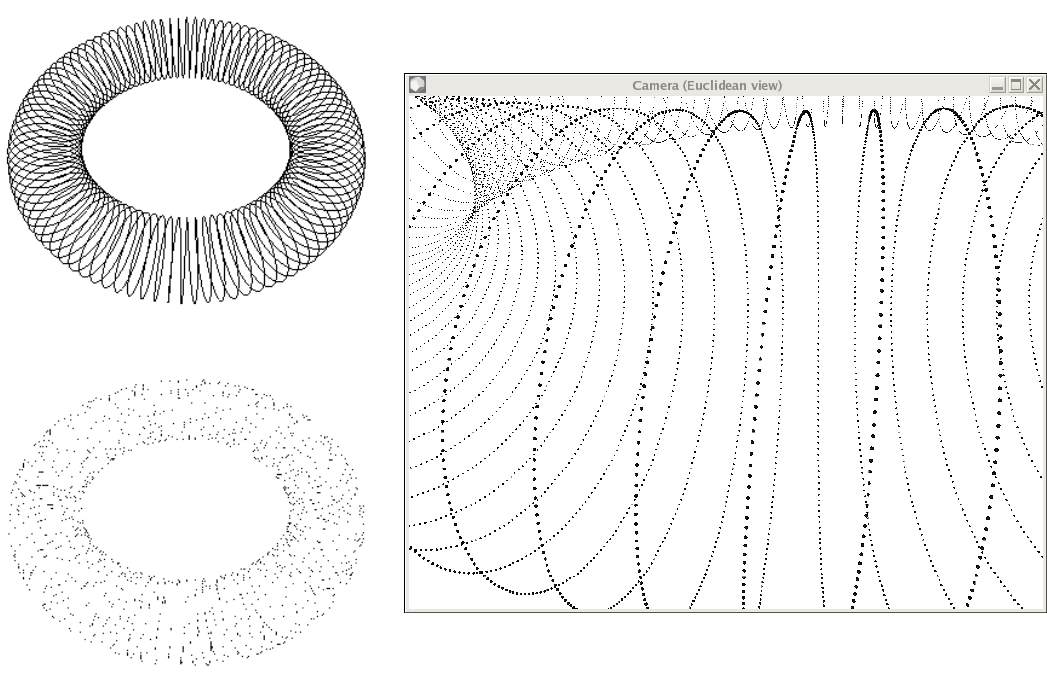}
\caption{Spiral on a torus}\label{fSpiral}
\end{figure}

\paragraph{Approximation}

We work from now on with a mass parameter $m$ such that $k=mn=5$. 
The persistence diagram of $\dmP$ is given in Figure~\ref{fPersDmP}:

\begin{figure}[!ht]
\centering
\includegraphics[height=10em]{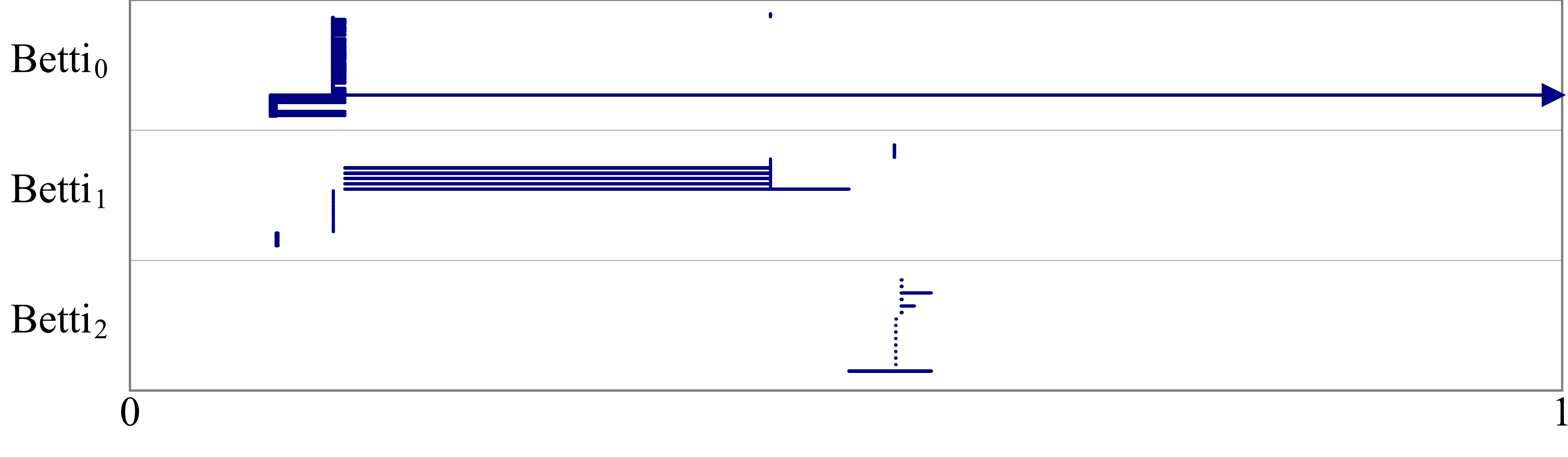}
\caption{$\Dgm{\dmP}$ for the cube skeleton with outliers with $k=5$}\label{fPersDmP}
\end{figure}

The diagrams obtained with our various approximations have very similar looks.
We only show the one obtained with the sparse Rips filtration with a parameter $\epsilon=0.5$ in Figure~\ref{fPersSpars}.

\begin{figure}[!ht]
\centering
\includegraphics[height=10em]{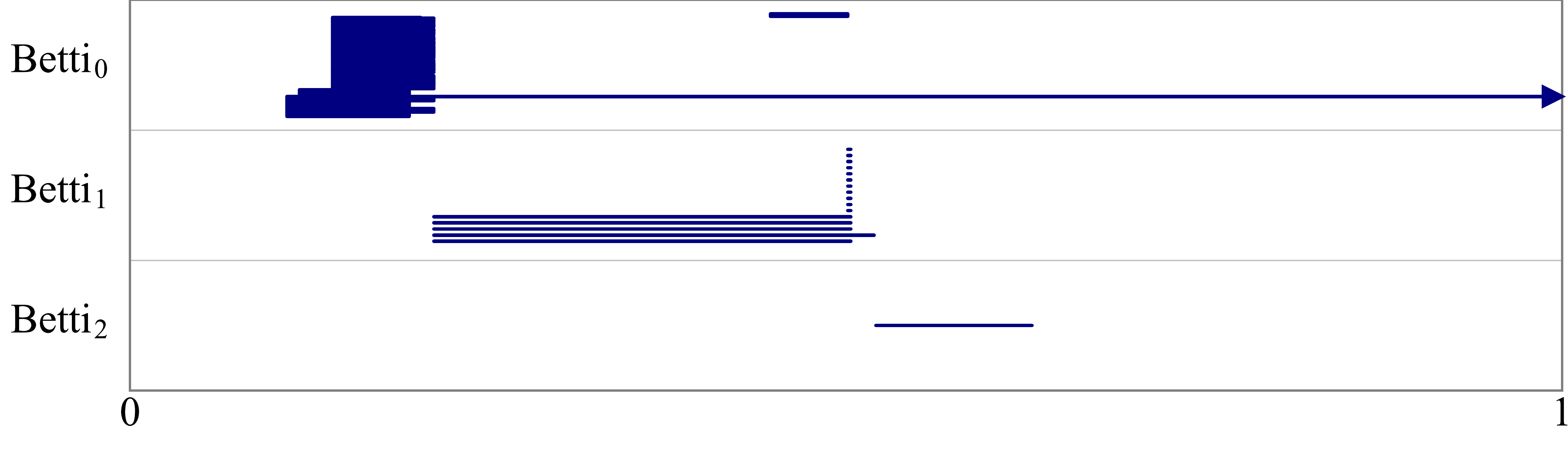}
\caption{$\Dgm{\{T_\alpha\}}$ for the cube skeleton with outliers with $k=5$ and $\epsilon=.5$}\label{fPersSpars}
\end{figure}

To compare diagrams, we use the bottleneck distances between the diagrams.
Figure~\ref{fMatrix} shows the distance matrix between the various diagrams, while Figure~\ref{fBottle} shows some bottleneck distances between persistence diagrams of different dimensions.
Note that $\Dgm{d_P}$ corresponds to the diagram obtained by using the distance function to the point cloud.

\begin{figure}[!ht]
\centering
\begin{tabular}{|c|cccccc|}
\hline
& $\Dgm{Skel}$ & $\Dgm{\dmP}$ & $\Dgm{\dPP}$ & $\Dgm{R_\alpha}$ & $\Dgm{T_\alpha}$ & $\Dgm{d_P}$ \\
\hline
$\Dgm{Skel}$ & 0 & .1528 & .1473 & .1473 & .1817 & .25 \\
$\Dgm{\dmP}$ & .1528 & 0 & .09872 & .0865 & .1183 & .2543 \\
$\Dgm{\dPP}$ & .1473 & .09872 & 0 & .0459 & .1084 & .2642 \\ 
$\Dgm{R_\alpha}$ & .1473 & .0865 & .0459 & 0 & .1128 & .2598 \\
$\Dgm{T_\alpha}$ & .1817 & .1183 & .1084 & .1128 & 0 & .2484 \\
$\Dgm{d_P}$ & .25 & .2543 & .2642 & .2598 & .2484 & 0 \\
\hline
\end{tabular}
\caption{Matrix of distances for the bottleneck distance}\label{fMatrix}
\end{figure}

\begin{figure}[!ht]
\centering
\begin{tabular}{|c|c|c|c|c|}
\hline
$\Dgm{A}$ & $\Dgm{B}$ & dim $0$ & dim $1$ & dim $2$ \\
\hline
$\Dgm{Skel}$ & $\Dgm{\dmP}$ & .05202 & .1528 & .1495\\
$\Dgm{\dmP}$ & $\Dgm{\dPP}$ & .09872 & .0195 &  .0972\\
$\Dgm{\dPP}$ & $\Dgm{R_\alpha(P,\dmP)}$ & .0007 & .0044 & .0459 \\
$\Dgm{R_\alpha(P,\dmP)}$ & $\Dgm{T_\alpha(P,\dmP)}$ & .0872 & .1128 & .0026 \\
\hline
$\Dgm{Skel}$ & $\Dgm{\dPP}$ & .0405 & .1473 & .0982 \\
$\Dgm{Skel}$ & $\Dgm{T_\alpha(P,\dmP)}$ & .1026 & .1817 & .098 \\
$\Dgm{Skel}$ & $\Dgm{d_P}$ & .25 & .2071 & .1481 \\
\hline
\end{tabular}
\caption{Bottleneck distances between diagrams}\label{fBottle}
\end{figure}

The largest difference is between $\Dgm{Skel}$ and $\Dgm{\dmP}$.
This is partly due to an effect of shifting while using the distance to a measure.
After this initial shift, the distance are small compared to the theoretical bounds. 
Notice that the different steps of the approximation do not have the same effect on all dimensions.

\begin{figure}[!ht]
\centering
\includegraphics[height=10em]{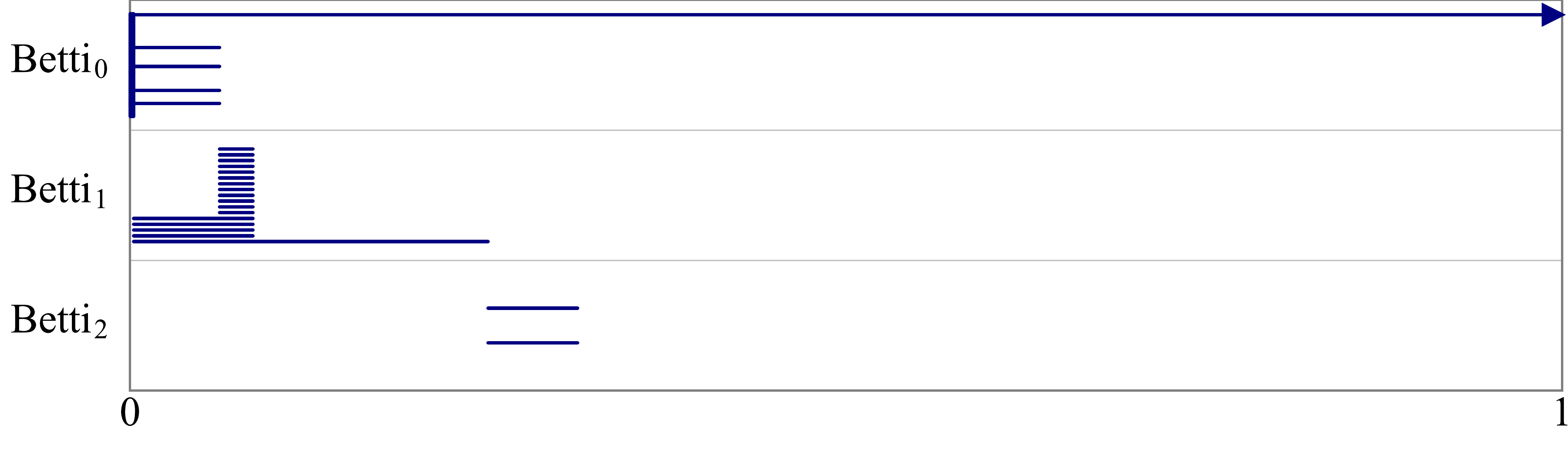}
\caption{$\Dgm{d_P}$ for the cube skeleton with outliers}\label{fDiagDP}
\end{figure}

All diagrams obtained by the different approximations are closer to $\Dgm{Skel}$ than the persistence diagram of the distance to the point cloud, $\Dgm{d_P}$ given in Figure~\ref{fDiagDP}.
For inference purposes, one crucial parameter is the \emph{signal-to-noise ratio}.
We define it as the ratio between the smallest lifespan of topological feature we aim to infer and the longest lifespan of noise features.
A ratio of $1$ corresponds to a signal that is not differentiable from the noise and $\infty$ corresponds to a noiseless diagram.
In our example, only the dimensions $1$ and $2$ are relevant as the dimension $0$ diagram corresponding to connected components has only one relevant feature and its lifespan is infinite.
Results are listed in Figure~\ref{fRatio}.

\begin{figure}[!ht]
\centering
\begin{tabular}{|c|c|c|}
\hline
Diagram & dim $1$ & dim $2$ \\
\hline
$\Dgm{Skel}$ & $\infty$ & $\infty$ \\
$\Dgm{\dmP}$ & $247$ &  $2.74$ \\
$\Dgm{\dPP}$ & $69.8$ & $43$ \\
$\Dgm{R_\alpha(P,\dmP)}$ & $\infty$ & $\infty$ \\
$\Dgm{T_\alpha(P,\dmP)}$ & $132$ & $\infty$ \\
$\Dgm{d_P}$ & $5.66$ & $1$ \\
\hline
\end{tabular}
\caption{Signal to noise ratios}\label{fRatio}
\end{figure}

Signal-to-noise ratios are clearly better than the one of $\Dgm{d_P}$.
Some of the approximation steps improve the ratio.
This is due to two phenomena.

When one goes from $\dmP$ to $\dPP$, the filtration eliminates the cells of the $k^{th}$ order Voronoi diagram that are far from the point cloud.
These cells induce local minima that produce noise features in the diagrams.
Removing them cleans parts of the diagram.
The same phenomenon happens with the witnessed $k$-distance perviously mentioned.

Using the Rips filtration instead of the \v{C}ech also reduces some noise.
It eliminates artifacts from simplices that are introduced and almost immediately killed in the \v{C}ech complex due to balls that intersect pairwise but have no common intersection.

\paragraph{Stability\\}

The weighted Rips filtration is stable with respect to noise. 
We illustrate this by studying the effect of an isotropic noise on our skeleton of a cube. 
We consider three different standard deviations for our noise.
Figure~\ref{fGauss} shows the bottleneck distances between the persistence diagram of the sparse weighted Rips structure with the Gaussian noise and the one without Gaussian noise.

\begin{figure}[!ht]
\centering
\begin{tabular}{|c|c|c|c|}
\hline
Standard deviation & $.05$ & $.1$ & $.5$ \\
\hline
$d_b$ in dimension $1$ & $.1469$ & $.2261$ & $.2722$ \\
$d_b$ in dimension $2$ & $.047$ & $.0914$ & $.1046$ \\
\hline
\end{tabular}
\caption{$d_b$ between $\Dgm{\{T_\alpha\}}$ with and without Gaussian noise}\label{fGauss}
\end{figure}

Unsurprisingly, the bottleneck distance is increasing with standard deviation of the noise.
The signal-to-noise ratio shown in Figure~\ref{fGaussRatio} is more interesting.

\begin{figure}[!ht]
\centering
\begin{tabular}{|c|c|c|c|c|}
\hline
Standard deviation & 0 & .05 & .1 & .5 \\
\hline
Ratio in dimension 1 & 132 & 8.27 & 3.17 & 1.04 \\
Ratio in dimension 2 & $\infty$ & $\infty$ & 100.2 & $\infty$ \\
\hline
\end{tabular}
\caption{Signal to noise ratio of $\Dgm{\{T_\alpha\}}$ depending on noise intensity}\label{fGaussRatio}
\end{figure}

Inferring correctly the homology of the cube skeleton is possible with standard deviation $0.05$ and $0.1$.
Figure~\ref{fPersGaus.1} shows the persistence diagram obtained with a standard deviation of $0.1$.
The $\infty$ in the $0.5$ case in dimension $2$ is not relevant as there is no noise but the feature is too small compared to the rest of the diagram as shown in Figure~\ref{fPersGaus.5}.
Note that $0.5$ corresponds to half of the side of the cube, and thus, it is logical to be unable to retrieve any useful information.

\begin{figure}[!ht]
\centering
\includegraphics[height=10em]{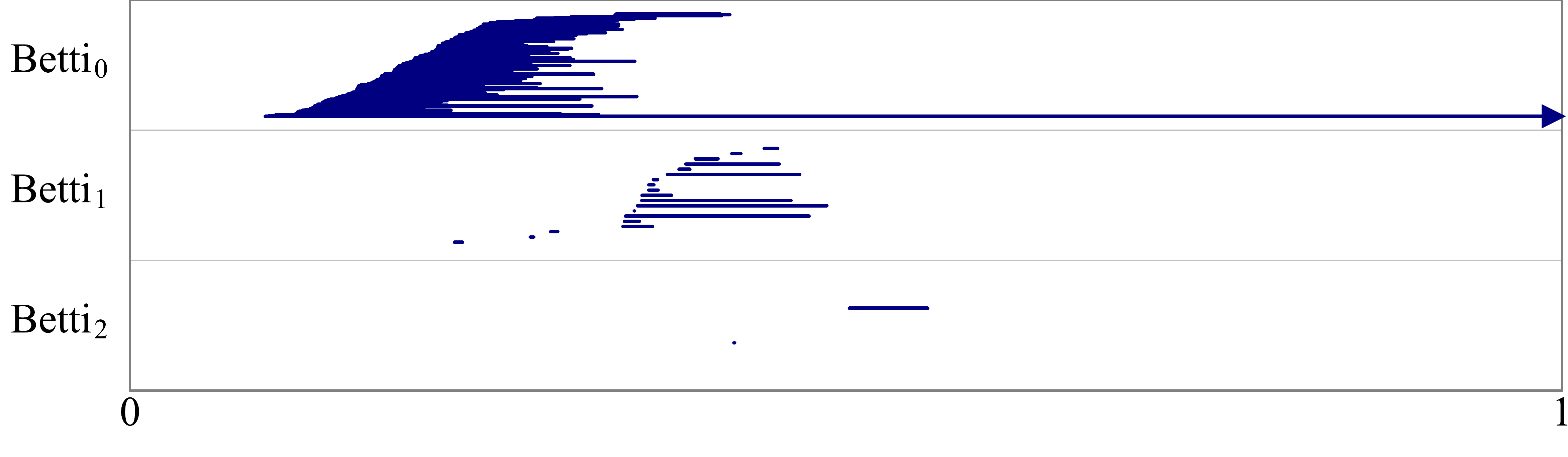}
\caption{Persistence diagram of $\{T_\alpha\}$ with $k=5$, $\epsilon=0.5$ and a Gaussian noise with standard deviation $0.1$}\label{fPersGaus.1}
\end{figure}

\begin{figure}[!ht]
\centering
\includegraphics[height=10em]{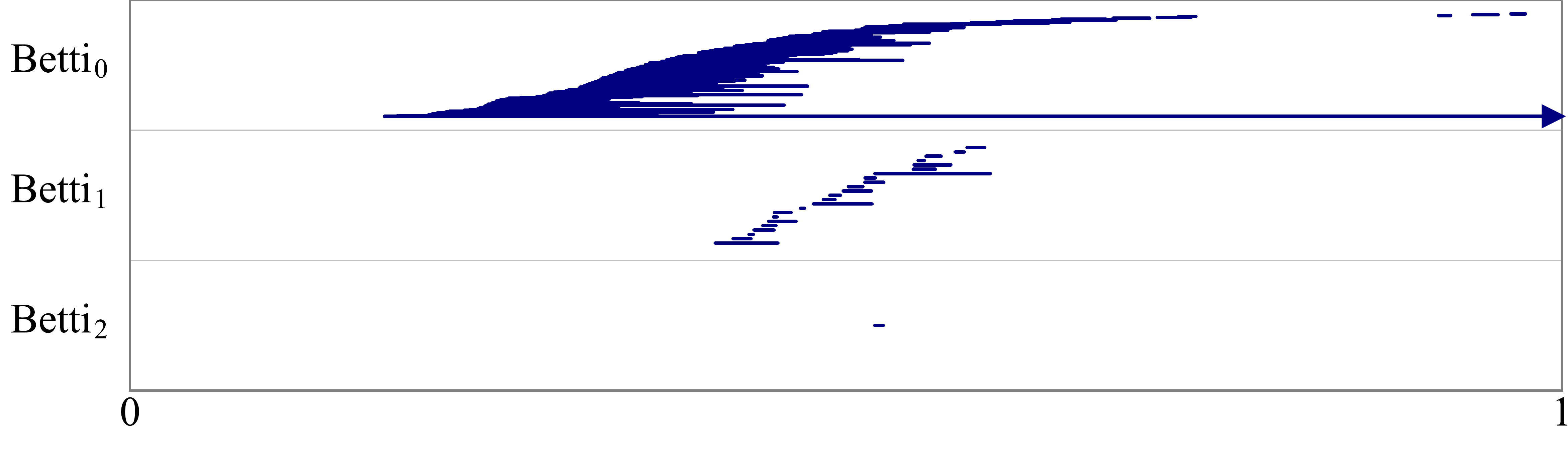}
\caption{Persistence diagram of $\{T_\alpha\}$ with $k=5$, $\epsilon=.5$ and a Gaussian noise with standard deviation $.5$}\label{fPersGaus.5}
\end{figure}

Some structure appears even with standard deviation as large as $0.5$. 
The three bigger features in dimension $1$ are relevant. 
However, we miss two elements and it is difficult to decide where to draw the frontier between relevant and irrelevant features.

\paragraph{Sparsification efficiency\\}
We introduced sparsification in Section~\ref{sec:rips_interleaving} to reduce the size of the Rips filtration.
The method introduced a new parameter $\epsilon$, and  the size of the filtration depends heavily on $\epsilon$.
The evolution of the size of the filtration depending on the parameter $\epsilon$ is given in Figure~\ref{fStatsSpiral}.

\begin{figure}[!ht]
\centering
\includegraphics[height=15em]{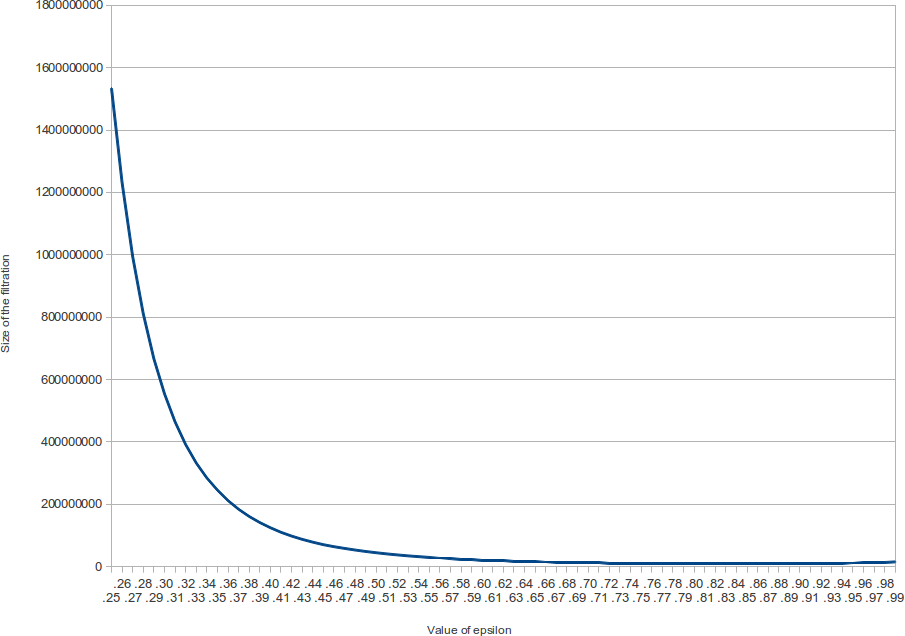}
\caption{Size of the filtration depending on $\epsilon$ for the spiral}\label{fStatsSpiral}
\end{figure}

The minimum size is reached around $\epsilon=.83$.
This minimum depends on the structure of the dataset.
For example, considering a set of points uniformly sampled in a square, we obtain decreasing size of the filtration.

The filtration size is nearly constant after a rapid decrease. 
In this example, the size is of order $10^7$ simplices for an input of $10^5$ vertices.
Computing persistent homology is tractable for any value in this range.
Structure in the data helps reduce the complexity of the sparse filtration.


\section{Conclusion} 
\label{sec:conclusion}

  In this paper, we generalize several aspects of the existing theory on the persistent homology of distances to measures from Euclidean space to general metric spaces.
  Then, we showed how to efficiently approximate the sublevels these distance functions with a linear number of metric balls.
  We gave a detailed analysis of the tightness of this approximation.
  
  We then showed how to give a sparse filtration that gives a guaranteed close approximation to the persistent homology of the distance to the measure.
  This last construction was given in the more general context of power distances.
  Thus, we have given a way to efficiently compute the persistent homology of the sublevel set filtration of any power distance function built on points in metric space of low doubling dimension.
  Since power distances can be used to approximate many different kinds of functions, we expect this technique will find many more uses in the future.
  
  A different perspective on our approach is that we use the sparse Rips filtration analogously to how one might use a grid in Euclidean space.
  It provides a structure over which one can go on to study many different functions.
  
  Lastly, we showed that this approach can be made practical, by providing some experimental results and analysis.


\bibliographystyle{plain}
\bibliography{sparserips}

\begin{thebibliography}{10}

\bibitem{carlsson09topology}
Gunnar Carlsson.
\newblock Topology and data.
\newblock {\em Bull. Amer. Math. Soc.}, 46:255--308, 2009.

\bibitem{ppmdCCGGO}
Fr{\'e}d{\'e}ric Chazal, David Cohen-Steiner, Marc Glisse, Leonidas~J. Guibas,
  and Steve~Y. Oudot.
\newblock Proximity of persistence modules and their diagrams.
\newblock In {\em Proceedings of the 25th Annual Symposium on Computational
  Geometry}, pages 237--246. ACM, 2009.

\bibitem{ghssspCCGMO}
Fr{\'e}d{\'e}ric Chazal, David Cohen-Steiner, Leonidas~J. Guibas, Facundo
  M{\'e}moli, and Steve~Y. Oudot.
\newblock Gromov-hausdorff stable signatures for shapes using persistence.
\newblock In {\em Computer Graphics Forum}, volume~28, pages 1393--1403. Wiley
  Online Library, 2009.

\bibitem{stcsesCCL}
Fr{\'e}d{\'e}ric Chazal, David Cohen-Steiner, and Andr{\'e} Lieutier.
\newblock A sampling theory for compact sets in {E}uclidean space.
\newblock {\em Discrete \& Computational Geometry}, 41(3):461--479, 2009.

\bibitem{gipmCCM}
Fr{\'e}d{\'e}ric Chazal, David Cohen-Steiner, and Quentin M{\'e}rigot.
\newblock Geometric inference for probability measures.
\newblock {\em Foundations of Computational Mathematics}, 11(6):733--751, 2011.

\bibitem{sspmCDGO}
Fr{\'e}d{\'e}ric Chazal, Vin de~Silva, Marc Glisse, and Steve Oudot.
\newblock The structure and stability of persistence modules.
\newblock {\em arXiv preprint arXiv:1207.3674}, 2012.

\bibitem{psgcCDO}
Fr{\'e}d{\'e}ric Chazal, Vin de~Silva, and Steve Oudot.
\newblock Persistence stability for geometric complexes.
\newblock {\em arXiv preprint arXiv:1207.3885}, 2012.

\bibitem{tpbresCO}
Fr{\'e}d{\'e}ric Chazal and Steve~Y. Oudot.
\newblock Towards persistence-based reconstruction in {E}uclidean spaces.
\newblock In {\em Proceedings of the twenty-fourth Annual Symposium on
  Computational Geometry}, pages 232--241. ACM, 2008.

\bibitem{arscgCS}
Kenneth~L. Clarkson and Peter~W. Shor.
\newblock Applications of random sampling in computational geometry, {II}.
\newblock {\em Discrete \& Computational Geometry}, 4(1):387--421, 1989.

\bibitem{spdCEH}
David Cohen-Steiner, Herbert Edelsbrunner, and John Harer.
\newblock Stability of persistence diagrams.
\newblock {\em Discrete \& Computational Geometry}, 37(1):103--120, 2007.

\bibitem{cgalAlphaShapes3D}
Tran Kai~Frank Da, S\'{e}bastien Loriot, and Mariette Yvinec.
\newblock {3D} alpha shapes.
\newblock In {\em {CGAL} User and Reference Manual}. {CGAL Editorial Board},
  {4.2} edition, 2013.
\newblock
  http\string://www.cgal.org/Manual/4.2/doc\_html/cgal\_manual/packages.html\#Pkg\string:AlphaShapes3.

\bibitem{ctpsmDFW}
Tamal~K. Dey, Fengtao Fan, and Yusu Wang.
\newblock Computing topological persistence for simplicial maps.
\newblock {\em arXiv preprint arXiv:1208.5018}, 2012.

\bibitem{ctaiEH}
Herbert Edelsbrunner and John~L. Harer.
\newblock {\em Computational topology: an introduction}.
\newblock American Mathematical Soc., 2010.

\bibitem{tpsELZ}
Herbert Edelsbrunner, David Letscher, and Afra Zomorodian.
\newblock Topological persistence and simplification.
\newblock In {\em Foundations of Computer Science, 2000. Proceedings. 41st
  Annual Symposium on}, pages 454--463. IEEE, 2000.

\bibitem{wkdGMM}
Leonidas Guibas, Dmitriy Morozov, and Quentin M{\'e}rigot.
\newblock Witnessed k-distance.
\newblock {\em Discrete \& Computational Geometry}, 49(1):22--45, 2013.

\bibitem{atH}
Allen Hatcher.
\newblock {\em Algebraic Topology}.
\newblock Cambridge University Press, 2002.

\bibitem{annMS}
David~M. Mount and Sunil Arya.
\newblock {ANN}: Library for approximate nearest neighbour searching.
\newblock 1998.

\bibitem{munkres84elements}
James~R. Munkres.
\newblock {\em Elements of Algebraic Topology}.
\newblock Addison-Wesley, 1984.

\bibitem{fhshcrsNSW}
Partha Niyogi, Stephen Smale, and Shmuel Weinberger.
\newblock Finding the homology of submanifolds with high confidence from random
  samples.
\newblock {\em Discrete \& Computational Geometry}, 39(1-3):419--441, 2008.

\bibitem{ZZZ}
Steve~Y. Oudot and Donald~R. Sheehy.
\newblock Zigzag zoology: Rips zigzags for homology inference.
\newblock In {\em Proceedings of the 29th annual Symposium on Computational
  Geometry}, pages 387--396, 2013.

\bibitem{lsavrfS}
Donald~R. Sheehy.
\newblock Linear-size approximations to the {V}ietoris-{R}ips filtration.
\newblock {\em Discrete \& Computational Geometry}, 49(4):778--796, 2013.

\bibitem{villani2003tot}
C.~Villani.
\newblock {\em {Topics in Optimal Transportation}}.
\newblock American Mathematical Society, 2003.

\bibitem{cphCZ}
Afra Zomorodian and Gunnar Carlsson.
\newblock Computing persistent homology.
\newblock {\em Discrete \& Computational Geometry}, 33(2):249--274, 2005.

\end{thebibliography}

\end{document}